\documentclass[acmsmall,screen,nonacm,review=false,timestamp=false]{acmart}
\usepackage{graphicx}
\usepackage{adjustbox}
\usepackage[disable]{todonotes}
\usepackage{booktabs}
\usepackage{cleveref}
\usepackage{tikz}
\usepackage{amsmath}
\usetikzlibrary{calc}
\usetikzlibrary {arrows.meta}
\usetikzlibrary{automata, positioning, arrows}
\usepackage[all,cmtip,2cell]{xy}
\usepackage{tikz-cd}
\usepackage{wrapfig}
\usepackage{caption}
\usepackage{subcaption}
\usepackage{nicefrac}
\usepackage{mdframed}

\usepackage{listings}
\usepackage{color}

\definecolor{dkgreen}{rgb}{0,0.6,0}
\definecolor{gray}{rgb}{0.5,0.5,0.5}
\definecolor{mauve}{rgb}{0.58,0,0.82}

\newcommand{\jr}[1]{\todo{JR: #1}}
\newcommand{\sj}[1]{\todo[inline]{SJ: #1}}

\newcommand{\Rom}[1]
    {\texorpdfstring{\MakeUppercase{\romannumeral #1}}{\romannumeral #1}}
    
\newcommand{\nat}{\mathbb{N}}

\newcommand{\natinf}{\mathbb{N}^{\infty}}

\newcommand{\real}{\mathbb{R}}
\newcommand{\nonegreal}{\real^{\infty}_{\geq 0 }}
\newcommand{\probinterval}{[0, 1]}

\newcommand{\settorf}{\mathbf{2}}
\newcommand{\sybsys}{S}
\newcommand{\sybspec}{R}
\newcommand{\sybprod}{S\otimes R}

\newcommand{\Frac}[2]{ {}^{#1} \! / \! {}_{#2} }
\newcommand{\lfpoint}[2][\cdot]{\llbracket #1 \rrbracket_{#2}}
\newcommand{\defeq}{\coloneq}

\newcommand{\cyset}[2][]{\mathrm{Cyl}^{#1}(#2)}
\newcommand{\pref}[1]{\mathrm{pre}(#1)}
\newcommand{\salg}[1]{\mathfrak{E}^{#1}}
\newcommand{\nset}[1]{[#1]}
\newcommand{\nsetbot}[1]{[#1]_{\bot}}
\newcommand{\supp}{\mathrm{supp}}
\newcommand{\dirac}[1]{\mathbf{1}_{#1}}

\newcommand\restr[2]{{
  \left.\kern-\nulldelimiterspace 
  #1 
  \vphantom{\big|} 
  \right|_{#2} 
  }}

\newcommand{\sets}{\mathbf{Set}}
\newcommand{\slice}[2]{\Frac{#1}{#2}}
\newcommand{\liftfunc}[2]{#1^{#2}}
\newcommand{\liftdist}[1]{\widetilde{#1}}
\newcommand{\predtran}[1]{\Phi_{#1}}

\newcommand{\id}{\mathrm{id}}

\newcommand{\fsys}[1]{#1_{\sybsys}}
\newcommand{\fspec}[1]{#1_{\sybspec}}
\newcommand{\fprod}[1]{#1_{\sybprod}}

\newcommand{\tvdsys}[1]{#1_{\sybsys}}
\newcommand{\tvdspec}[1]{#1_{\sybspec}}
\newcommand{\tvdprod}[1]{#1_{\sybprod}}

\newcommand{\modsys}[1]{#1_{\sybsys}}
\newcommand{\modspec}[1]{#1_{\sybspec}}
\newcommand{\modprod}[1]{#1_{\sybprod}}
\newcommand{\coalg}[1]{\mathrm{CoAlg}(#1)}

\newcommand{\accepting}{F}

\newcommand{\powersetfunc}{\mathcal{P}}
\newcommand{\fpowersetfunc}{\mathcal{P}^{f}}
\newcommand{\distributionfunc}{\mathcal{D}}

\newcommand{\subdistributionfunc}{\distributionfunc_{\leq 1}}

\newcommand{\suptlflagfunc}[2]{\powersetfunc\big((#1+\settorf)\times #2\big)}
\newcommand{\mcflagfunc}[1]{\distributionfunc\big(#1+\settorf\big)}
\newcommand{\dmcflagfunc}[1]{\distributionfunc\big(#1 + \{\star\}\big)}

\newcommand{\wmealyfunc}[1]{\big(\powersetfunc(#1\times \settorf\times \nat)\big)^A}

\newcommand{\ltsfuncsemiring}[3][A]{\powersetfunc\big((#2 + \{\star\})\times #1\times #3\big)}

\newcommand{\mcfunc}[2][A]{\distributionfunc(#2+\{\star\})\times{#1}}

\newcommand{\mcinffunc}[2][A]{\distributionfunc(#2)\times{#1}}

\newcommand{\dfautomata}[2][A]{(#2\times \settorf)^{#1}}

\newcommand{\exepnfautomata}[2][A]{\big(\powersetfunc(#2\times \settorf)\big)^{#1} }

\newcommand{\nempfaclang}[1][A]{\powersetfunc({#1}^{+})}

\newcommand{\faclang}[1][A]{\powersetfunc({#1}^{\ast})}

\newcommand{\fsublang}[1][A]{\subdistributionfunc({#1}^{+})}

\newcommand{\probmeasure}{\mathbb{P}}

\newcommand{\mrg}[2]{\restr{#1}{#2}}

\newcommand{\recharge}{\textsf{\textcolor{yellow!50!black}{recharge}}}
\newcommand{\arid}{\textcolor{brown!50!black}{\textsf{arid}}}
\newcommand{\lake}{\textcolor{blue!50!black}{\textsf{lake}}}
\newcommand{\volcano}{\textcolor{red!50!black}{\textsf{volcano}}}
\newcommand{\sand}{\textsf{sand}}

\newcommand{\yellowWord}[1]{\textcolor{orange}{#1}}
\newcommand{\cbquery}{q_{cb}}

\newif\ifdraft\drafttrue

\newif\iffull\fulltrue 

\newif\ifmaterial\materialfalse 

\AtEndPreamble{%
   \theoremstyle{acmdefinition}
   \newtheorem{remark}[theorem]{Remark}}
\renewcommand{\cref}[1]{\Cref{#1}}
\creflabelformat{enumi}{#2(#1)#3}
\crefname{theorem}{Thm.}{Theorems}
\crefname{definition}{Def.}{Defs}
\crefname{proposition}{Prop.}{Props}
\Crefname{equation}{Eq.}{Eqs}
\crefname{equation}{Eq.}{Eqs}
\crefname{lemma}{Lem.}{Lemmas}
\crefname{remark}{Rem.}{Remarks}
\crefname{example}{Ex.}{Examples}
\crefname{proof}{Proof.}{Proofs}
\crefname{appendix}{Appendix}{Appendixes}
\crefformat{section}{{\S}#2#1#3}
\crefname{figure}{Fig.}{Figs}
\Crefname{equation}{}{}
\AtBeginDocument{%
  \providecommand\BibTeX{{%
    \normalfont B\kern-0.5em{\scshape i\kern-0.25em b}\kern-0.8em\TeX}}}

\setcopyright{acmcopyright}
\copyrightyear{2024}
\acmYear{2024}
\acmDOI{XXXXXXX.XXXXXXX}

\acmJournal{PACMPL}
\acmVolume{x}
\acmNumber{x}
\acmArticle{x}
\acmMonth{7}

\begin{document}

\title{A Unifying Approach to Product Constructions for Quantitative Temporal Inference}

\author{Kazuki Watanabe}
\email{kazukiwatanabe@nii.ac.jp}
\affiliation{%
   \institution{National Institute of Informatics}
   \state{Tokyo}
   \country{Japan}
}
\affiliation{%
   \institution{The Graduate University for Advanced Studies (SOKENDAI)}
   \state{Tokyo}
   \country{Japan}
}
\author{Sebastian Junges}
\email{sebastian.junges@ru.nl}
\affiliation{%
   \institution{Radboud University}
   \state{Nijmegen}
   \country{The Netherlands}
}
\author{Jurriaan Rot}
\email{jrot@cs.ru.nl }
\affiliation{%
   \institution{Radboud University}
   \state{Nijmegen}
   \country{The Netherlands}
}
\author{Ichiro Hasuo}
\email{i.hasuo@acm.org }
\affiliation{%
   \institution{National Institute of Informatics}
   \state{Tokyo}
   \country{Japan}
}
\affiliation{%
   \institution{The Graduate University for Advanced Studies (SOKENDAI)}
   \state{Tokyo}
   \country{Japan}
}
\begin{abstract}
Probabilistic programs are a powerful and convenient approach to formalise distributions over system executions. A classical verification problem for probabilistic programs is \emph{temporal inference}: to compute the likelihood that the execution traces satisfy a given temporal property. This paper presents a general framework for temporal inference, which applies to a rich variety of quantitative models including those that arise in the operational semantics of probabilistic and weighted programs. 

The key idea underlying our framework is that in a variety of existing approaches, the main construction that enables temporal inference is that of a \emph{product} between the system of interest and the temporal property. We provide a unifying mathematical definition of product constructions, enabled by the realisation that 1) both systems and temporal properties can be modelled as \emph{coalgebras} and 2) product constructions are \emph{distributive laws} in this context. Our categorical framework leads us to our main contribution: a sufficient condition for \emph{correctness}, which is precisely what enables to use the product construction for temporal inference. 

We show that our framework can be instantiated to naturally recover a number of disparate approaches from the literature including, e.g., partial expected rewards in Markov reward models, resource-sensitive reachability analysis, and weighted optimization problems.
Further, we demonstrate a product of weighted programs and weighted temporal properties as a new instance to show the scalability of our approach.

\end{abstract}

\ccsdesc[500]{Software and its engineering~Formal software verification}

\keywords{temporal inference, coalgebra, probabilistic programming}


\maketitle

\ifmaterial
\else

\section{Introduction}
Probabilistic programming is a powerful methodology that uses syntax and semantics from programming languages to describe probabilistic models and the distributions they induce. Probabilistic programming ecosystems make probabilistic reasoning accessible to a wide audience~\cite{DBLP:conf/icse/GordonHNR14} and make it easier to integrate probabilistic reasoning as a function inside a program~\cite{GoodmanMRBT08}. 

Classically, probabilistic programs describe distributions over the output of a program. Many probabilistic programming languages come with (semi-)automatic inference algorithms that effectively compute statistics of these distributions. 
A range of design decisions balance the expressivity of the languages with the performance of specific types of inference. 
 Many approaches execute the programs, efficiently sampling paths and collecting the outputs (see e.g.~\cite{abs-1809-10756,barthe2020foundations}). Exact alternatives typically reflect a weighted model counting approach whose performance relies on the effective aggregation of such paths~\cite{DBLP:conf/se/FilieriPV14,holtzen2020scaling,DBLP:journals/pacmpl/SusagLHR22}. 
 Program calculi for verification of probabilistic programs instead rely on marginalizing the paths out and reasoning using fuzzy state predicates~\cite{DBLP:journals/jcss/Kozen85,DBLP:series/mcs/McIverM05,DBLP:phd/dnb/Kaminski19}.

\begin{wrapfigure}[17]{r}{0pt}
    \centering
    \includegraphics[width=.6\textwidth]{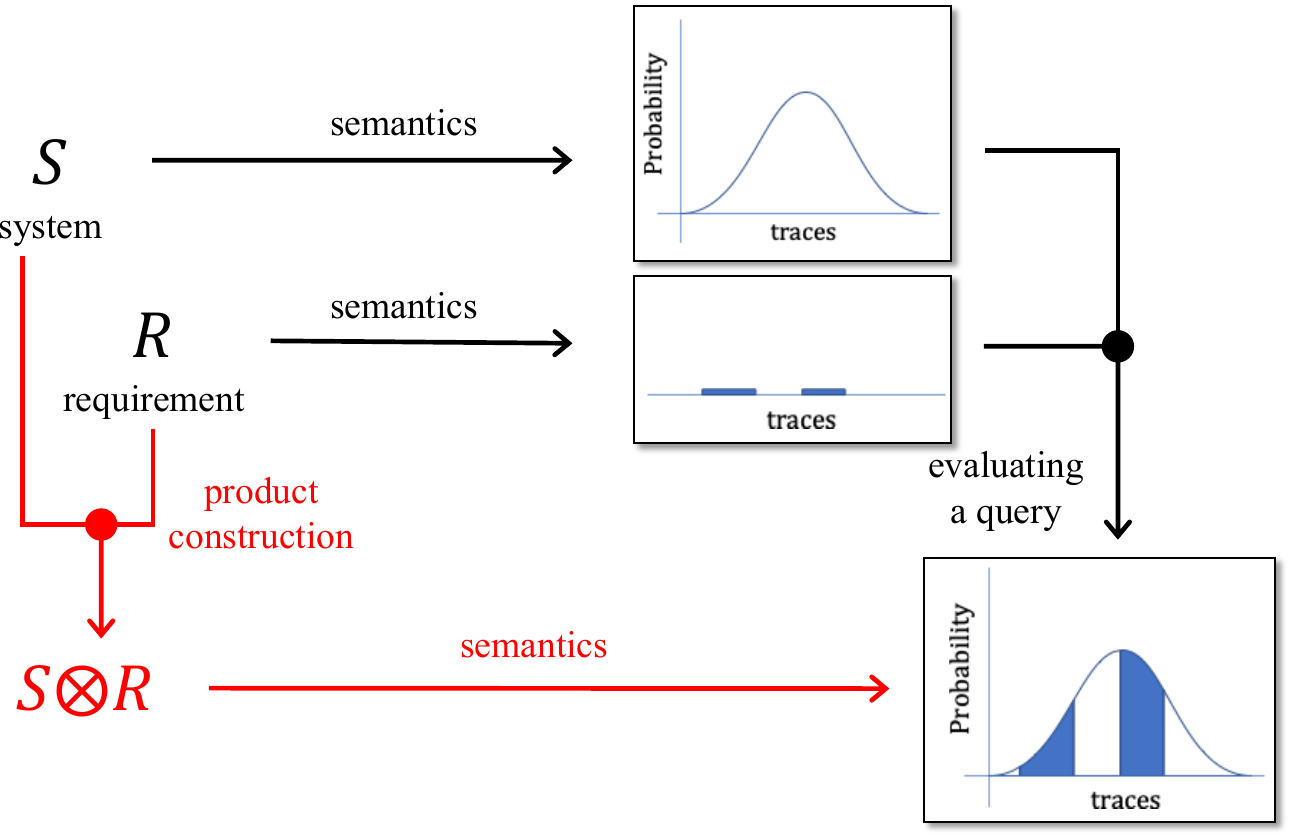}


        \caption{Temporal inference: its problem definition (black) and its tractable solution by product construction (red).}
    \label{fig:inference_prob_agent_operation}
\end{wrapfigure}

\paragraph{From output distributions to traces}
In many situations not only the output distribution is relevant, but also the observable behaviour during execution.
In these settings, probabilistic programs describe distributions over system \emph{traces}, which can be captured as formal languages or time series. 
Indeed, probabilistic programs are useful to describe the sequential executions of, e.g., network protocols~\cite{SmolkaKFK017,DBLP:conf/pldi/SmolkaKKFHK019}, computer hardware~\cite{DBLP:conf/fmics/RobertsLHBRCZ21}, and planning problems~\cite{sanner2010relational}.
Given a \emph{requirement}, which describes the desirable traces in a formal language, the problem of \emph{temporal inference} is then to determine the probability of program traces to meet the requirement (\cref{fig:inference_prob_agent_operation},  black).

It is common to specify the requirement as a state machine, which can then be used as a monitor for the requirement to solve the temporal inference problem. However, marginalization of paths to the requirement is not trivial here and a particular challenge of temporal inference in this setting is that naively, the set of system traces is gigantic due to the combinatorial explosion.

\paragraph{Product constructions}
Fortunately, even in this setting most temporal inference problems do not actually require analyzing (an approximation of) the full distribution of traces described by the program. Suppose, for instance, that we wish to infer the probability of two successive packets in a network being dropped. For this purpose it suffices to monitor whether the previous packet has been dropped and whether at some point two successive packets have been dropped. All other events that occur during traces can be ignored, leading to a massive reduction of the space of traces.

More generally, to answer a temporal inference query, it suffices to infer the (distribution of the) final status of the monitor, and use this to answer the original query.
This powerful idea of analyzing a synchronous composition or \emph{product} of a system and its monitor is illustrated in \cref{fig:inference_prob_agent_operation} (red); we illustrate an example in~\cref{sec:agent_operation}. It is  a fundamental idea in the verification of linear temporal logics~\cite{Pnueli77,DBLP:journals/jcss/VardiW86} and has been reapplied for a variety of other properties in quantitative verification, such as resource-bounded reachability~\cite{AndovaHK03,LaroussinieS05,DBLP:conf/nfm/BaierDDKK14,HartmannsJKQ20}, timed traces on continuous-time Markov chains~\cite{DBLP:conf/lics/ChenHKM09}, and conditional probabilities~\cite{DBLP:conf/tacas/BaierKKM14}.  
These constructions are derived on a case-by-case basis. Our aim is a generic framework that describes product constructions and sufficient criteria that state for what types of queries inference on these product constructions is correct.

In this work, we assume that operationally, systems are captured by a weighted or probabilistic transition system and the requirement by finite automata. All these types of transitions systems are instances of \emph{coalgebras}~\cite{Jacobs16,Rutten00}, which we therefore choose to represent both systems and requirements. At this level of generality,  we define a synchronous product of the operational models for (any) system $S$ and requirement $R$.
The key idea here is commutation: it does not matter whether we first take the semantics of $S$ and $R$ and adequately combine them, or first take the product of $S$ and $R$ and compute the semantics for this product. The main theorem in this paper formalizes the requirements for which these approaches indeed commute.

\paragraph{A broad framework}
Our framework applies to both well-understood and to less studied programs and queries. For probabilistic programs, we support almost-surely terminating programs that describe a distribution over finite traces, and regular specifications captured by DFAs. 
We also use this construction to support expected rewards, cost-bounded specifications that accumulate cost along every path.
 Importantly, the restriction to almost-surely terminating programs can be lifted, and for such programs we support safety properties.
Furthermore, our construction is not restricted to probabilistic programs. Indeed, temporal inference on a variety of weighted programs is supported and allows inference on travel costs and optimal resource utilization. 
We present the list of instances in~\cref{fig:list_of_instances}, where we highlight that we discuss different types of programs, types of requirements, and that these lead to different types of inferences on the product.

\begin{figure}
    \scalebox{0.7}{
    \begin{tabular}{lccc}\toprule
    Section & System & Requirement & Inference on Product\\\midrule
    \cref{sec:coalgebraic_inference},\cref{sec:coal_model_checking_regular_semantics}     &  Probabilistic Program that Almost-Surely Terminates                        & RSP              & Terminating Probability\\
    \cref{sec:rewards_prob_program}                                                &  Probabilistic Program with Reward               & RSP              & Partial Expected Reward \\
    \cref{sec:cost_bounded_consumption}                                            &  Probabilistic Program that Almost-Surely Terminates                          & Resource                             & Terminating Probability\\
    \cref{sec:MCsNeverTerminate}                                                   &  Probabilistic Program that Never Terminates     & RSP              & Terminating Probability\\
    \cref{sec:weighted_programming}                                                &  Weighted Program that Terminates                               & RSP              & Shortest Path\\
    \cref{sec:quantitativeInferenceWeightedSystemRequirement}                      &  Weighted Program that Terminates                                & Weighted RSP   & Shortest Path
    \\\bottomrule
    \end{tabular}
    }
    \caption{Instances of coalgebraic inference with coalgebraic product construction. RSP denotes the regular safety property.
    }
    \label{fig:list_of_instances}
\end{figure}

\paragraph{The framework as a stepping stone to \emph{product programs}.}
While the framework first and foremost provides a theoretical underpinning, we briefly reflect on its role in developing efficient inference engines.
The only approaches to temporal inference that do not rely on a (possibly ad-hoc) product construction are sampling-based: They sample traces from the system $S$ and then check whether that trace meets the requirement $R$. Such approaches cannot be employed for efficient exact inference. We firmly believe that efficient implementations will need to implement a product construction. 
Indeed, it is generally beneficial to have a small highly efficient and potentially formally verified core that enables fast inference. Reducing temporal inference to such efficient core computations via a product construction is, in our opinion, \emph{the} way to provide tool support for temporal inference in probabilistic and weighted programs. 
In fact, probabilistic inference tools such as  Psi~\cite{GehrMV16} or Dice~\cite{holtzen2020scaling} can be used for temporal inference; however, this typically requires human ingenuity, basically encoding the product construction into the input. As we show in this paper, the correctness of such product constructions is not always trivial and the manual construction of such products remains error prone.
The notable exception to this manual translation is the \textsc{Rubicon} transpiler~\cite{DBLP:conf/cav/HoltzenJVMSB20}, which exemplifies the creation of such product programs to infer finite-horizon reachability in Dice~\cite{holtzen2020scaling} programs. Indeed, the framework in this paper can been as a step toward generalizing this approach: The product construction in this paper is seen as providing an operational semantics of some \emph{product program}, i.e., a program that incorporates the specification and adjusts its return type (e.g.~\cite{DBLP:conf/cav/HoltzenJVMSB20}). These programs are amenable to any approximate or exact inference engine and the correctness of the construction is independent of the correctness of the inference engine.

\paragraph{Algorithmic implications of the framework.}
Finally, we want to remark that the framework provides two \emph{algorithmic} approaches for temporal inference.
On one hand, it provides, as a byproduct of our general product construction,  a least fixed point characterization of the solution of a temporal inference problem. This characterization provides value-iteration-like~\cite{DBLP:books/wi/Puterman94} iterative algorithms to approximate the solution. On the other hand, this fixed point characterization is a stepping stone to apply ideas such as inductive invariants~\cite{DBLP:conf/cav/KoriUKSH22,DBLP:conf/cav/ChatterjeeGMZ22,DBLP:conf/tacas/BatzCJKKM23} to answer a variety of queries.

\paragraph{Contributions.}
To summarize, this paper presents a unifying framework for probabilistic and weighted temporal inference with temporal properties. Concretely, we contribute
\begin{itemize}
    \item a generic definition of temporal inference queries (in~\cref{sec:coalgebraic_inference}),
    \item a generic approach to performing this type of inference, based on a coalgebraic product construction (in~\cref{sec:coal_model_checking_regular_semantics}), 
 together with a correctness criterion  for this product construction, and
    \item case studies of coalgebraic quantitative inference of variable complexity (in~\cref{sec:cost_bounded_consumption,sec:rewards_prob_program,sec:weighted_programming,sec:quantitativeInferenceWeightedSystemRequirement,sec:MCsNeverTerminate}). 
\end{itemize}
We use the correctness criterion (and a weaker correctness criterion) to recover results on established (\cref{sec:cost_bounded_consumption,sec:rewards_prob_program,sec:weighted_programming,sec:MCsNeverTerminate}) and new (\cref{sec:quantitativeInferenceWeightedSystemRequirement}) temporal inference queries. We start the paper with an overview that illustrates the various types of temporal inference and the product construction. We present related work in \cref{sec:relatedwork}.


\paragraph{Notation}
We write $\settorf$ for the set $\{\top, \bot\}$ of Boolean values. We write $\bot$ for the least element of a partial order, assuming it exists. 
Let $X$ be a set. The set $\powersetfunc(X)$ contains the subsets of $X$ and the set $\fpowersetfunc(X)$ the finite subsets of $X$. 
We write $\distributionfunc(X)$ for the set of distributions on $X$ whose support is at most countable,
and $\subdistributionfunc(X)$ for the set of subdistributions on $X$ with countable support.
For a subdistribution $\nu$, we write $\supp(\nu)$ for the set of supports of $\nu$. We write $\natinf$ for the set of integers and the positive infinity $\infty$.  For a natural number $n\in \nat$, we let $\nset{n} := \{1, \dots, n\}$. We write $(X \rightarrow Y)$ or $Y^X$ for the set of functions from $X$ to $Y$. For functions of the form $f \colon X \rightarrow (Y \rightarrow Z)$ we often write $f(x,y)$ instead of $f(x)(y)$, and similarly for functions of the form $P \colon X \rightarrow \distributionfunc(Y)$ we write $P(x,y)$ instead of $P(x)(y)$.
We write $\dirac{f=g}$ for the value $1$ if $f = g$ holds, $0$ otherwise. 
We write $X+Y$ or $X\uplus Y$ for the disjoint union of the sets $X$ and $Y$.

\section{Overview}
\label{sec:overview}

The typical inference task in probabilistic programming languages is to infer (statistics of) the posterior distribution described by the program. 
In this overview, we highlight probabilistic and weighted programs where we perform inference not on the outcome of the program, but on the executions themselves, which we refer to as \emph{temporal inference}, and demonstrate how we can use \emph{product constructions} to reduce the problem to  inference on the posterior.

These examples of temporal inference problems motivate the overall challenge addressed in this paper: to give a formal framework that captures the mathematical essence of such product constructions and their use for temporal inference. Moreover, this framework should be flexible and general enough to easily show the correctness of product constructions for a broad range of systems and temporal properties, including the various examples in this overview section.

\subsection{Probabilistic Temporal Inference in Probabilistic Programs}
\label{sec:agent_operation}

\begin{figure}
    \centering
    \begin{minipage}{4cm}
        \vspace{-1em}
        \includegraphics[scale=0.23]{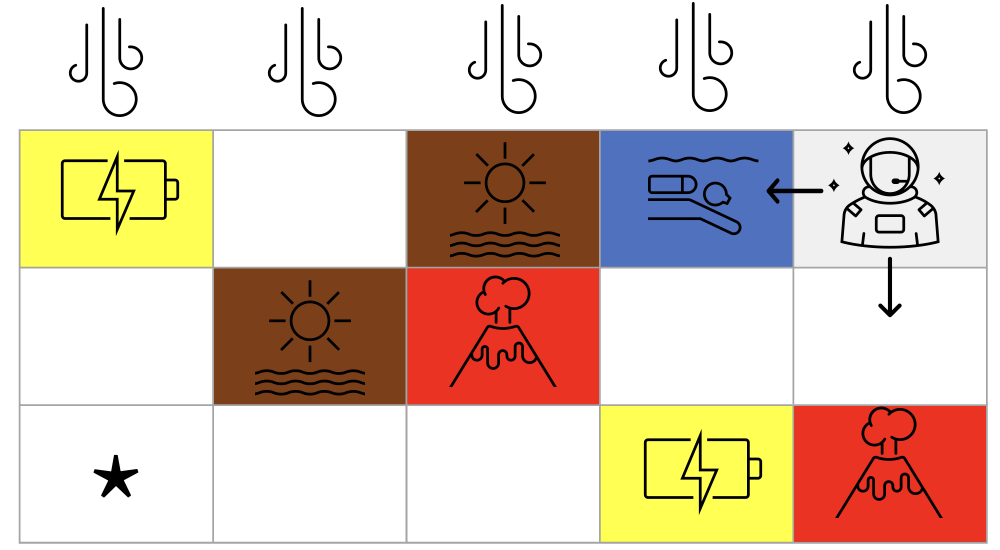}
    \end{minipage}
    \hspace{5pt}
    \begin{minipage}{5cm}
   
    \begin{lstlisting}
i = 5; j = 3; a = sand;
while (i > 1 or j > 1) {
  {i <- max(i-1, 1)} [4/5] 
  {j <- max(j-1, 1)} 
  a <- label(i,j);
}
    \end{lstlisting}
    
    \end{minipage}
    \begin{minipage}[t]{0.3\textwidth}
  \centering
  \raisebox{-20pt}{
    \scalebox{0.7}{
    \begin{tikzpicture}
        \node[state] (q0) at (0, 0) {$y_0$};
        \node[state, accepting] (q1) at (2.5, 0) {$y_1$};
        \node[state] (q2) at (-2.5, 0) {$y_2$};
        \node[state] (q3) at (0, -2) {$y_3$};
        \draw[->] (q0) edge[above] node{\recharge} (q1);
        \draw[->] (q0) edge[below,bend left,] node{\lake} (q2);
         \draw[<-] (q0) edge[above right,bend right] node{\arid} (q2);
         \draw[->] (q1) edge[above, bend right] node{\lake} (q2.north);
         \draw[->] (q2) edge[below left, bend right] node[align=center]{\recharge,\\\volcano} (q3);
         \draw[->] (q0) edge[right] node{\volcano} (q3);
         \draw[->] (q1) edge[below right, bend left] node{\volcano} (q3);
    \end{tikzpicture}  
    }
    }
  \end{minipage}
        \caption{Robot moves in the $5\times 3$ grid world. The  requirement of the robot is given by the DFA. The state $y_0$ is the initial state, and $y_1$ is the unique accepting state in the DFA. We omit all self-loops for readability. 
        For instance, we omit the self-loops on the states $y_2$ and $y_3$ whose assigned character is $\lake$.} 
    \label{fig:agent_operating}
\end{figure}

We start with a classical reach-avoid inference task, which we concretize with a robot in a small grid world, inspired by~\cite{Vazquez-Chanlatte18,Vazquez-Chanlatte21}. The grid world consists of cells,  as displayed in \cref{fig:agent_operating}. A cell contains either \recharge{} stations (yellow), \lake{s} (blue), \arid{} areas (brown), or \volcano{s} (red). All other cells contain nothing but \sand{} (white). As the robot moves, we observe the labels of the cell, which we denote $A = \{ \sand, \recharge, \lake, \arid, \volcano \}$. Any finite sequence of areas is thus representable in $A^{+}$. We consider the case where the robot follows a fixed plan that describes in every cell what direction it aims to go. However, the actual outcome of trying to move in that direction is probabilistic, due to, e.g., wind gusts from the north.\footnote{More realistic examples in probabilistic robots show that the outcome of an actuator action is often uncertain.}

The (simplified) probabilistic program described in \cref{fig:agent_operating} is a possible implementation of the robot movement. The robot starts at $\langle i=5, j=3 \rangle$, which is a sand cell. 
As long as the robot has not reached the lower left corner, it will move left with probability $\frac{4}{5}$ and move down otherwise. We explicitly store the label of the current cell with a function $\texttt{label}$ which returns this label. We remark that this program almost-surely terminates --- we consider programs that do not necessarily almost-surely terminate later in this section.  The program then induces a distribution $\nu\in \distributionfunc(A^{+})$ over traces which represent paths that the robot takes. \jr{how does this distribution look? maybe give the probability of one or two traces. I guess we should also say that it terminates in the bottom-left more explicitly}

The aim of the robot is to reach a charging station, but this is subject to some rules: the robot should never travel over volcanos; a wet robot should not recharge; a robot becomes wet when travelling through lakes; and a robot becomes dry when travelling through arid areas. 
Formally, these rules can be captured as a set $L \subseteq A^{+}$ of safe paths, which we refer to as the \emph{requirement}. Such requirements can be encoded using regular expressions or using temporal logics such as (fragments of) (finitary) LTL.
In this example, we assume that the requirement $L$ is given by a deterministic finite automaton (DFA) as in~\cref{fig:agent_operating}. 

Generally, we want to find the probability that a trace generated by the program belongs to the set of traces specified in the requirement. More precisely, let $\nu \in \distributionfunc(A^{+})$ be the trace distribution generated by the program and $L \subseteq A^+$ (or equivalently, $L \in \nempfaclang$) be the language represented by the DFA for the requirement. 
Formally, our aim is to answer a \emph{(temporal) inference query}, namely to compute $q_\mathsf{prob}(\nu, L)$ in:
\[ q_\mathsf{prob}\colon  \distributionfunc(A^{+})\times \nempfaclang\rightarrow \probinterval, \quad\text{where}\quad\! q_\mathsf{prob}(\nu, L) \defeq \sum_{w\in L} \nu(w). \]
Such inference queries are heavily studied in the area of probabilistic model checking~\cite{DBLP:reference/mc/BaierAFK18}.\footnote{Model checking would typically define a \emph{models relation} $\models$ such that, e.g., $\nu \models L$ if $q(\nu, L) \geq \frac{1}{2}$.} \jr{do we need this footnote?}
Towards efficient inference, the key aim is to avoid computing $\nu$ and $L$ first, as the set $A^{+}$ is (countably) infinite. Instead, the standard approach is to construct a \emph{product of the operational semantics of the program and the requirement}, given as a Markov chain (MC) and a DFA, respectively.

We give concrete (but simplified) operational semantics\footnote{The operational semantics of probabilistic programs can be defined as a Markov chain that includes the assignment to the variables but also the program counter, see e.g., \cite{DBLP:journals/pe/GretzKM14}.
We omit the program counter for simplicity.} for the program in Fig.~\ref{fig:agent_operating}. The program halts in a dedicated state $\star$. 
The operational semantics is an MC with state labels, referred to below as the \emph{system}. Formally, it is defined using \emph{(regular) states} $X = \{ ( i, j ) \mid  i \in \nset{1, 5}, j \in \nset{1, 3}, (i > 1\text{ or } j > 1) \}$ where $\nset{1, m}$ is the set $\{1, \dots, m\}$, \emph{terminating state} $\star$, \emph{labels} $A = \{ \sand, \recharge, \lake, \arid, \volcano \}$ assigned to the states by the \emph{labelling function} $l\colon X\rightarrow A$, \emph{initial state} $(5, 3)$, and the \emph{transition probability} $P\colon X \rightarrow \distributionfunc(X + \{ \star\})$ with for $i > 1$ or $j > 1$,   
\begin{align*}
    P\big( (i, j) \big)\big((i', j') \big) &\defeq  \begin{cases}\frac{4}{5} &\text{ if $i' = \mathrm{max}(i - 1,\, 1)$ and $j' =j $,}\\
     \frac{1}{5} &\text{ if $i'=i$ and $j' = \mathrm{max}(j - 1,\, 1)$,}\\
     0 &\text{ otherwise.}
     \end{cases}
\end{align*}
Here, we regard the terminating state $\star$ as the grid $(1, 1)$ in order to make the definition simpler. 

Similarly, the DFA shown in~\cref{fig:agent_operating}, consists of states $Y = \{ y_0, y_1, y_2, y_3 \}$, a set of accepting states $\accepting = \{ y_3 \}$ and a labelled transition relation $T \subseteq Y \times A \times Y$ as indicated by the arrows in~\cref{fig:agent_operating}.

From the system and the  requirement, the product construction yields  a transition system whose states are pairs of states of the system and the  requirement, and whose transitions are ``synchronized'' transitions. In this product, the labels do not explicitly appear anymore. 
Concretely, the product of the MC and the DFA given above is again a Markov chain, with two semantically different terminal states, i.e., the state space is $X \times Y + \settorf$. States are thus either a tuple of an MC state and DFA state, or a Boolean flag,
which is used to indicate whether upon reaching $\star$ the execution is accepted by the DFA. For the original transition matrix $P$, we can define the new transition matrix $P'$, as follows:
\begin{align*}
    P'\big( (x, y),\,(x',  y')\big) &\defeq \begin{cases}
       P\big(x,\, x' \big) &\text{ if } (y, l(x),y') \in T,\\
       0 &\text{ otherwise,}
    \end{cases}\\
    P'\big((x, y),\, [y' \in F] \big) &\defeq 
       P\big(x,\,\star\big),\text{ where $(y, l(x), y')\in T$.}
\end{align*}
Here $[y' \in F]$ is an instance of the following notation: for a predicate $\mathrm{pred}\colon Y\rightarrow \settorf$, we let $[\mathrm{pred}(y)] = \top$ iff $\mathrm{pred}(y)$ evaluates to true.
	   
The temporal inference problem for the robot can now be answered by asking the probability of terminating while accepting, i.e., of reaching the state $\top$ or equivalently of returning $\top$. 
This probability coincides with $q(\nu, L)$ by construction as it keeps track of the path for reaching $\star$ is in $L$. Thus, the product construction enables \emph{temporal} inference by applying inference methods that analyze the distribution of outputs from the product.

\vspace{3mm}
\begin{mdframed}
The essence of tractable solutions to (exact) temporal inference is to construct a synchronous product of the program and the requirement. 
\end{mdframed}
\vspace{3mm}

\subsection{Extended Probabilistic Temporal Inference Queries}\label{ssec:ext-inference}
There exist plenty of inference queries where the requirement or query changes.

\paragraph{Partial expected rewards}
 Almost classical are \emph{partial expected rewards} (rewards are also often referred to as costs)~\cite{Baier08,Baier0KW17}. They have been studied, for instance, to determine the expected runtime of randomized algorithms~\cite{KaminskiKMO18}. 
 In terms of the earlier example, we may annotate every state with additional movement time that depends, e.g., on the type of surface. 
 We describe the system as an MC with integer movement cost. 
A natural query is to ask what the \emph{expected arrival time} is upon reaching a recharge station. 
 In that case, the requirement remains the same DFA as above, which accepts the safe paths including a recharge station. 
As before, a query takes the system and the requirement, but is now defined to compute the \emph{weighted} sum of probability times reward. We continue this discussion in \cref{sec:rewards_prob_program}.

\paragraph{Reward-bounded reachability}
Rather than inferring the \emph{expected} time to reach a cell, it is often relevant to determine the probability to reach a cell before some deadline. 
This is particularly relevant for limited resources, such as battery levels, where we want to compute the probability to reach a recharge station before the battery is empty or, in a different scenario, to compute the probability to reach the airport \emph{before} a flight departs. These scenarios  are instances of \emph{cost-bounded reachability} problems~\cite{AndovaHK03,DBLP:conf/nfm/BaierDDKK14,HartmannsJKQ20}.
 Intuitively, every path gets a finite resource cost, and we are interested in analyzing the probability that we obtain a path that meets a threshold on the cost of this path. Here, the main observation that is necessary is to realize that this requirement can be captured by finite automata. 

\paragraph{Conditional probabilities}
We want to highlight that the provided queries (e.g. $q_\mathsf{prob}$ from~\cref{sec:agent_operation}) can be building blocks towards even more generic queries. 
For example, conditional probability queries help to restrict a model to the behavior that corresponds to a set of observations. They can be supported using Bayes' rule. Formally, the following query supports analyzing the probability of satisfying a  requirement $L$ conditioned on satisfying a conditional observation $C$.
\[ q_{\textsf{cond}}\colon  \distributionfunc(A^{+})\times \nempfaclang \times \nempfaclang \rightarrow \probinterval, \quad q_{\textsf{cond}}(\nu, L, C) \defeq   \frac{q_\mathsf{prob}(\nu, L \cap C)}{q_\mathsf{prob}(\nu, C)}  = \frac{\sum_{w\in L \cap C} \nu(w)}{\sum_{w\in C} \nu(w)}. \] 
The query can be used to define the probability that we reach a charging station conditioned on reaching an arid field. The partial expected rewards can also be used to compute conditional rewards, which will be clarified in \cref{sec:rewards_prob_program}.

\begin{remark}
We assume here that the conditioning is given separately from the system. Many probabilistic programs allow to explicitly formulate the conditioning inside the program using \texttt{observe} statements~\cite{NoriHRS14,HurNRS14,DBLP:journals/toplas/OlmedoGJKKM18}\todo{should be better citations here}. However, such programs can be seen as formulating both a distribution over paths (without conditioning) and a set of paths to use for conditioning. 
\end{remark}

\subsection{Probabilistic Temporal Inference for Programs that do not Terminate}

\begin{figure}
    \centering
    \begin{minipage}{8cm}
   
    \begin{lstlisting}
x = 5; y = 3; a = sand;
while (true) {
  {x <- max(x-1,0)}   [1/4] {y <- max(y-1,0)} [1/4]
    {x <- min(x+1,5)} [1/4] {y <- min(y+1,3)};
  a <- label(x,y);
}
    \end{lstlisting}
    
    \end{minipage}
    \begin{minipage}[t]{0.3\textwidth}
  \centering
  \raisebox{-20pt}{
    \scalebox{0.7}{
    \begin{tikzpicture}
        \node[state] (q0) at (0, 0) {$y_0$};
        \node[state, accepting] (q1) at (3, 0) {$y_1$};
        \node[state] (q2) at (3, -2) {$y_2$};
        
        \draw[->] (q0) edge[above] node{\recharge} (q1);
        
        \draw[->] (q0) edge[above] node[rotate=325] {\volcano} (q2);
        
        \draw[->] (q0) edge[loop above] node {$A \setminus \{ \recharge, \volcano \}$} (q0);
        
        \draw[->] (q1) edge[loop right] node{$A$} (q1);
        
        \draw[->] (q2) edge[loop right] node{$A$} (q2);
    \end{tikzpicture}  
    }
    }
  \end{minipage}
        \caption{Reactive gridworld: Robot moves in the $5\times 3$ grid world in any direction, the program never halts. The  requirement expresses that we want to reach a charging station without reaching a volcano. }
    \label{fig:reactive_agent_operating}
\end{figure}

Above, we considered programs that almost-surely terminate. However, it is often useful to describe models that do not terminate. For instance, the program in~\cref{fig:reactive_agent_operating} never halts: The robot continuously moves throughout the grid.  Still, it is natural to ask for the probability of eventually reaching a charging station without reaching a volcano before.
 We give this  requirement using a simple DFA that captures the (finite, unbounded) executions that satisfy the  requirement. 
 We note that we are thus interested in paths that have a finite (but unbounded) \emph{prefix} in the requirement, i.e., we consider regular safety properties. 
 
 Since the traces of the robot are now infinite, we need (limited) measure-theoretic tools, in particular, the cylinder set construction~\cite{Baier08}. Consequentially, the query no longer takes a distribution over finite paths, but a specific $\sigma$-algebra as first argument. However, while the ingredients become more involved, the generic framework from before still applies. We formalize the construction in~\cref{sec:MCsNeverTerminate}. We also show how we can express the inference in programs that almost-surely terminate to programs that never halt.

\subsection{Inference Queries for Weighted Programs} 
\label{sec:ski_snowboard}

We now consider a type of quantitative temporal inference different from the probabilistic examples.
\emph{Weighted Programming}~\cite{CohenSS11,BrunelGMZ14,GaboardiKOS21,BatzGKKW22,BelleR20} is an emerging paradigm that aims to model mathematical problems using syntax and semantics from programming languages. 
Probabilistic programs are a special instance of \emph{weighted programs} where  we \emph{multiply} probabilities along paths and  then \emph{sum} over all paths. Weighted programs allow to parameterize the operations along paths and over all paths. 

 One instance of weighted programs is formed by optimization problems that attempt to minimize the costs expressed as weights along traces. Formally, these program use the tropical semiring $(\natinf = \nat + \{\infty\}, \min, +)$ to combine minimization and addition over the extended natural numbers: the program minimizes over all traces while summing the costs along the traces. 
We are interested in temporal inference that encodes a constrained optimization variant, in which we minimize only over traces that satisfy a given requirement. 

\begin{figure}
    \centering
    \begin{subfigure}[b]{0.6\textwidth}
    \centering
        \begin{minipage}{7cm}
        \begin{lstlisting}[mathescape=true]
        t: town; nt: town, a: action;
        t = "hometown";
        $\yellowWord{\text{min}}${
          while (t != "destination") { 
            (nt, a) in possible_choice(t);
            $\yellowWord{\text{add}}${t <- move(nt, a)};
          }
        }
        \end{lstlisting}
        \end{minipage}
        \caption{Traveler's move. }
        \label{fig:tranfer}
    \end{subfigure}
    \begin{subfigure}[b]{0.35\textwidth}
    \centering
    \begin{minipage}{5cm}
    \raisebox{30pt}{
    \adjustbox{scale=0.8,center}{
    \begin{tikzpicture}
        \node[state] (q0) at (0, 0) {$y_0$};
        \node[state,accepting] (q1) at (3, 0) {$y_1$};
        \draw[->] (q0) edge[above] node{T} (q1);
        \draw[->] (q0) edge[loop above] node{P,\, B,\, T} ();
    \end{tikzpicture}
    }
    }
    \end{minipage}
    \caption{Requirement $(\dagger)$.}
    \label{fig:spec_transfer}
    \end{subfigure}
    \caption{Travelling problem. The functions $\yellowWord{\mathrm{min}}$ and $\yellowWord{\mathrm{add}}$ in the weighted program are the primary operations on the costs (or weights in general). The state $y_0$ is the initial state, and $y_1$ is the unique accepting state in the NFA $d$. The characters $\mathrm{P},\mathrm{B},\mathrm{T}$ correspond to taking a plane, a bus, and a train, respectively.}
    \label{fig:travelling_problem}
\end{figure}

We exemplify the problem statement with a simple planning problem: the \emph{Travelling Problem}. A scientist is about to leave home to travel to a conference using a variety of potential transportation means. The weighted program in~\cref{fig:tranfer} models the transfer of the scientist; traces of the program correspond to transfer histories, and fares along traces are accumulated as total costs. 
The scientist can provide a requirement with personal preferences. The scientist loves taking trains and wants to arrive by train at the destination  $(\dagger)$. The nondeterministic finite automaton (NFA) in~\cref{fig:travelling_problem} captures this simple requirement. 

The travelling problem can be formulated as the following temporal inference query: let (i) $T\in \powersetfunc(A^{+}\times \nat)$ be the set of pairs $(w, m)$ of a trace $w\in A^{+}$ and a possible accumulated cost $m\in \nat$ along the trace $w$, and (ii) $L\in \nempfaclang$ be the recognized language of the NFA.
The inference query 
\begin{align*}
    q\colon \powersetfunc(A^{+}\times \nat)\times \nempfaclang\rightarrow  \natinf, \quad q(T, L)\defeq  \min \{m \mid w\in L,\, (w, m)\in T \}.
\end{align*}
naturally formalizes the requested quantity. However, computing the value  $q(T, L)$ is not trivial as the sets $A^{+}$ and $\nat$ are countably infinite. 

The key to solving the Travelling Problem again lies in the construction of a product of the operational semantics of the program and the requirement, in a similar style to~\cref{sec:agent_operation}. The operational semantics of the weighted program shown in~\cref{fig:tranfer} is a weighted transition system, 
equipped with (i) \emph{states} $X$ for the set of towns that can be used to transfer; (ii) the \emph{destination} $\star$; (iii) \emph{labels} $A$ for the set of means of transport; and (iv) \emph{weights} $\nset{M} =\{1,\dots, M\}$ for the set of costs. 
For each town $x\in X$, the non-deterministic transition from $x$ specifies a set of triples $(x', a, m)$ where $x'$ is the next town or the destination $\star$, $a$ is its means of transport, and $m$ is the associated cost.
Requirements are captured by NFAs.

Then, the product of the operational semantics of the program and the requirement is the ``synchronized'' transition system.
The Travelling Problem can then be solved as a simple shortest path problem on the product. Again, the product construction leads to an efficient solution by disregarding the individual steps. In~\cref{sec:weighted_programming} we show how weighted temporal inference fits into our framework.

In the Traveling Problem above, the costs are completely determined by the system. In fact, it is possible to add additional costs via the requirement. That can be used to penalize certain travel options and allows for additional flexibility for the user to specify constraints. We show such a generalization in \cref{sec:quantitativeInferenceWeightedSystemRequirement}.

\subsection{Towards a Unifying Theory of Product Constructions}

The range of quantitative temporal inference problems discussed above are broad, but share a common pattern: there is an underlying product construction which we may use to efficiently solve the original temporal inference problem. The challenge we address is to identify the mathematical essence of product constructions and their role in practical solutions for temporal inference.
More precisely, let $\lfpoint[\cdot]{c}$ be the semantics of the system $c$ (e.g., an MC with labels), $\lfpoint[\cdot]{d}$ the semantics of the requirement $d$ (e.g., a DFA or an NFA),
and let $\lfpoint[\cdot]{c\otimes d}$ of the product $c\otimes d$ (e.g., an MC without labels); and let $q$ be the temporal inference query of interest.
We aim for the following equality that states the ``correctness'' of the product for temporal inference:
\begin{equation}
\label{eq:correct_overview}
    \lfpoint[x,y]{c\otimes d} = q(\lfpoint[x]{c}\times \lfpoint[y]{d}).
\end{equation}
This equality says that solving the temporal inference problem at hand is equivalent to solving the product $c\otimes d$ of the system $c$ and the requirement $d$.
The left hand side of the equation is the semantics of this product. 
Importantly, the semantics of the product itself should be independent of the query $q$, and this is the reason why we can directly compute the semantics of the product for the temporal inference, without computing the semantics of the system $c$ and the requirement $d$.
The equality~\eqref{eq:correct_overview} is the starting point for our main question:
\vspace{3mm}
\begin{mdframed}
How can we uniformly construct a product $c\otimes d$ from a system $c$ and a requirement $d$ and ensure that Eq.~\cref{eq:correct_overview} holds?  
\end{mdframed}
\vspace{3mm}
Towards an answer, we use \emph{coalgebras}~(see, e.g.,~\cite{Jacobs16}) as the foundation for our theory of temporal inference and product constructions. The theory of coalgebras provides us with a generic expression of transition systems that covers a wide variety of system types (including MCs, DFAs, NFAs, Mealy machines, MCs with costs, and weighted systems) and a way of defining their semantics as least fixed points of suitable predicate transformers.

The structural approach offered by coalgebrals enables the general treatment of quantitative temporal inference. 
In fact, we observe that a product between system and requirement is induced by a \emph{distributive law} between functors---these have been widely studied in, for instance, the theory of programming languages (e.g.~\cite{TuriP97,AguirreB23,GoncharovMSTU23}), and automata theory (e.g.~\cite{Jacobs06,DBLP:journals/corr/abs-1302-1046,Jacobs0S15,KlinR16}). 
This formulation naturally leads us to a definition of correctness, formalising Eq.~\cref{eq:correct_overview}.
Our main result is a sufficient condition for correctness, referred to as our \emph{correctness criterion} (\cref{thm:correctness_product}) and its weaker variant (\cref{prop:weaker_correctness_criterion}). 
As we show in the second part of this paper, our framework is powerful enough to cover a variety of temporal inference problems, including all examples discussed above.

\section{Coalgebraic Inference}
\label{sec:coalgebraic_inference}

In this section, we introduce a generic framework, called \emph{coalgebraic inference}, which  unifies the temporal inference problems shown in the previous section.
This framework relies on the generality of coalgebras as a mathematical notion of state-based system: we model both systems and requirements as coalgebras,
and define the notion of \emph{query} in this context.

\subsection{Generic Semantics of Coalgebras}
Intuitively, a \emph{coalgebra} maps \emph{states} $X$ to $F(X)$, which is a set that specifies (i)~the type of transition, e.g., probabilistic or non-deterministic, and (ii)~the information of a state $x \in X$, e.g., a label $a$ or a weight $m\in \nat$ assigned to $x$.
For instance, in this paper $F(X)$ is $\distributionfunc{(X+\{\star\})}\times A$ for a labeled MC with the target state $\star$, and $F(X)$ is $\dfautomata{Y}$ for a DFA.
Explanation on how coalgebras of these types then correspond to labeled MCs and DFAs follows below in Example~\ref{ex:mc_moore_mealy}.
We require that $F$ is a \emph{functor} (on the category of sets). 
\begin{definition}[functor]
\label{def:functor}
An \emph{endofunctor} $F$ (on the category of sets) maps sets $X$ to sets $F(X)$ and functions $f\colon X\rightarrow Y$ to functions $F(f)\colon F(X)\rightarrow F(Y)$, such that: 
\begin{itemize}
    \item $F$ preserves identities, i.e., for any set $X$, $F(\id_X) = \id_{F(X)}$,
    \item $F$ preserves composition, i.e., for functions $f\colon X\rightarrow Y$ and $g\colon Y\rightarrow Z$, $F(g\circ f) = F(g)\circ F(f)$.  
\end{itemize}
Functors are a more general concept in category theory~\cite{MacLane2}, but we only need endofunctors on the category of sets. We therefore refer to endofunctors on
the category of sets, as defined above, as \emph{functors} for simplicity. 
\end{definition}

\begin{definition}[coalgebra]
\label{def:coalgebra}
A \emph{coalgebra} for a functor $F$ and states $X$ is a function $c\colon X\rightarrow F(X)$. 
\end{definition}

\begin{example}[coalgebras]
\label{ex:mc_moore_mealy}
We list some coalgebras discussed in this paper:
    \begin{enumerate}
        \item A \emph{(labeled) Markov chain (MC) with a target state $\star$} is a pair of functions $(P,\, l)$, where $P\colon X \rightarrow \distributionfunc{(X+\{\star\})}$ captures the probabilistic transitions, and $l\colon X\rightarrow A$ is a labelling of states. A Markov chain $(P,\, l)$ is indeed a coalgebra $c\colon  X\rightarrow\mcfunc{X}$ given by $c(x) \defeq \big(P(x),\, l(x)\big)$ for each $x\in X$.
        
        The underlying functor $F$ is defined on a set $X$ by 
        $F(X) = \mcfunc{X}$. For a function $f \colon X \rightarrow Y$, the function $F(f) \colon F(X) \rightarrow F(Y)$ informally replaces all occurrences of elements $x \in X$ by $f(x)$, yielding again a distribution. More precisely, 
        $F(f)(\nu_X,a) = (\nu_Y,a) \in \mcfunc{Y}$ where $\nu_Y(y) \defeq \sum_{x\in f^{-1}(y)} \nu_X(x)$, and $\nu_Y(\star) \defeq \nu_X(\star)$.
    \item  A \emph{deterministic finite automaton (DFA)} is a coalgebra $d\colon Y\rightarrow \dfautomata{Y}$, where $Y$ is the set of states, and $A$ are the labels. Note that output is on transitions, as in Mealy machines; for each state $y\in Y$ and label $a\in A$, the value $b\in \settorf$ given by $(y', b) \defeq d(y)(a)$ is $\top$ if $y'$ is an accepting state, and $\bot$ otherwise. 
    \item A \emph{Markov chain without the target state $\star$} is a coalgebra  $c\colon X\rightarrow \mcinffunc{X}$.
    \item A \emph{weighted transition system} is a coalgebra $c\colon X\rightarrow \ltsfuncsemiring{X}{\nat}$, where each non-deterministic transition provides a pair $(a, m)\in A\times \nat$ of the label $a$ and the weight $m$. 
    \end{enumerate}
 
\end{example}

We use a generic notion of semantics of coalgebras, characterised by a \emph{semantic domain} and a \emph{modality}. 
Examples of semantic domains in this paper are sub-distributions over traces, and languages over a fixed alphabet. We illustrate the modalities in examples below.
\begin{definition}[semantic structure, predicate transformer]
\label{def:sem_coal}
A \emph{semantic structure} for a functor $F$ is a tuple of  an $\omega$-complete partial order ($\omega$-cpo)\footnote{ See~\cref{sec:omitted_definitions} for the definition of $\omega$-cpos. It is standard to assume $\omega$-cpos for coalgebraic semantics; see e.g.~\cite{HasuoJS07}.} $\mathbf{\Omega} \defeq (\Omega, \preceq)$ and a \emph{modality} $\tau\colon F(\Omega)\rightarrow \Omega$.\footnote{Modalities are defined on the category of sets, hence we implicitly apply the forgetful functor to $\Omega$.} We  refer to an element $v\in \Omega$ as a \emph{semantic value}.
Given a coalgebra $c \colon X \rightarrow F(X)$, the induced \emph{predicate transformer} $\Phi_{\tau,c} \colon \Omega^{X}\rightarrow \Omega^{X}$ is given by $\Phi_{\tau,c}(u) \defeq \tau\circ F(u)\circ c$.
If $\tau$ and $c$ are clear from the context we write $\Phi$ instead of $\Phi_{\tau,c}$. 
\end{definition}

Inducing predicate transformers by modalities is a standard technique in coalgebraic modal logic (e.g.~\cite{KupkeP11}). 
Informally, modalities explain how to combine semantic values of successor states into a single semantic value of the current state. 
For instance, for an unlabeled MC, we can take a semantic domain $\Omega$ as $\probinterval$. 
The reachability probability from the current state is computed by a weighted sum over the reachability probabilities from the successor states; this is modeled by the modality $\tau$. 
In a labeled MC, the semantic domain is the set $\fsublang$ of subdistributions over traces; we show this in~\cref{ex:tau-semanticsMC}.  
Throughout the paper, we only consider $\omega$-continuous predicate transformers $\Phi$. A sufficient condition for the $\omega$-continuity of $\Phi$ is the continuity of the modality $\tau$, that is, $\tau\circ F(\bigvee u_n) = \bigvee \tau\circ  F(u_n)$ for any $\omega$-chain $(u_n)_{n\in \nat}$ in $\Omega^X$ with a suitable $\omega$-cpo structure on $\Omega^{FX}$, and the continuity of the sequential composition with the underlying coalgebra $c$. 

\begin{definition}[semantics of coalgebras]
Let $\Phi_{\tau,c}$ as above be an ($\omega$-continuous) predicate transformer.
The \emph{semantics of the coalgebra $c$}
is the least fixed point $\lfpoint{c}\colon X\rightarrow \Omega$ of $\Phi_{\tau,c}$ in Eq.~\cref{eq:semantics}. 
\begin{equation}
\label{eq:semantics}
\begin{tikzcd}
        F(X)  \arrow[rr, "F(\lfpoint{c})"] \arrow[drr,phantom,description,"=_{\mu}"] & & F(\Omega) \arrow[d, "\tau"]  \\
        X \arrow[u, "c"] \arrow[rr, "\lfpoint{c}"]& & \Omega 
\end{tikzcd}
\end{equation}
\end{definition}

As $\Phi$ is $\omega$-continuous on $\omega$-cpo $\mathbf{\Omega}$, the least fixed point $\lfpoint{c}$ exists, and can be computed as the join $\bigvee_{n\in \nat}\predtran{}^{n}(\bot)$, by the Kleene fixed point theorem~\cite{cousot1979constructive,Baranga91}. 

\begin{figure}
    \centering
    \begin{minipage}{9cm}
    \centering
    \begin{tikzcd}
        \mcfunc{X}  \arrow[rrr, "\mcfunc{\lfpoint{c}}"] \arrow[drrr,phantom,description,"=_{\mu}"]  &&& \mcfunc{\fsublang} \arrow[d, "\tau_{S}"] \\
        X \arrow[u, "c"] \arrow[rrr, "\lfpoint{c}"]& && \fsublang
    \end{tikzcd}
    \end{minipage}

    \caption{The semantics $\lfpoint{c}$ of a Markov chain (MC) $c\colon  X\rightarrow \mcfunc{X}$, 
    where $\Phi(u)  \defeq  \tau_{S}\circ \mcfunc{u} \circ c$. }
    \label{fig:coalgebraicTSwithfibration}
\end{figure}

\begin{figure}
    \centering
    \begin{minipage}[t]{0.65\textwidth}
    \scalebox{0.7}{
    \begin{tikzpicture}
        \node[draw, rectangle, rounded corners, initial, initial text=] (s0) at (0, 0) {$x_0, \sand$};
        \node[draw, rectangle, rounded corners] (s1) at (2, 1) {$x_1, \lake$};
        \node[draw, rectangle, rounded corners] (s2) at (2, -1) {$x_2, \sand$};
        \node[draw, rectangle, rounded corners] (s3) at (4.5, 1) {$x_3, \recharge$};
        \node[draw, rectangle, rounded corners] (s4) at (4.5, -1) {$x_4,\volcano$};
        \node[state] (s5) at (7, 0) {$\star$};
        \draw[->] (s0) edge[above] node{$\frac{4}{5}$} (s1);
        \draw[->] (s0) edge[above] node{$\frac{1}{5}$} (s2);
        \draw[->] (s1) edge[above] node{$1$} (s3);
        \draw[->] (s2) edge[above] node{$\frac{4}{5}$} (s3);
        \draw[->] (s2) edge[above] node{$\frac{1}{5}$} (s4);
        \draw[->] (s3) edge[above] node{$1$} (s5);
        \draw[->] (s4) edge[above] node{$1$} (s5);
    \end{tikzpicture}  
    }
    \end{minipage}
    \begin{minipage}[t]{0.3\textwidth}
        \centering
          \scalebox{0.7}{
          \begin{tikzpicture}
              \node[state] (q0) at (0, 0) {$y_0$};
              \node[state, accepting] (q1) at (2.5, 0) {$y_1$};
              \node[state] (q2) at (-2.5, 0) {$y_2$};
              \node[state] (q3) at (0, -2) {$y_3$};
              \draw[->] (q0) edge[above] node{\recharge} (q1);
              \draw[->] (q0) edge[below,bend left,] node{\lake} (q2);
               \draw[<-] (q0) edge[above right,bend right] node{\arid} (q2);
               \draw[->] (q1) edge[above, bend right] node{\lake} (q2.north);
               \draw[->] (q2) edge[below left, bend right] node[align=center]{\recharge,\\\volcano} (q3);
               \draw[->] (q0) edge[right] node{\volcano} (q3);
               \draw[->] (q1) edge[below right, bend left] node{\volcano} (q3);
          \end{tikzpicture}  
          }
    \end{minipage}
    \caption{A (small) MC (on the left) and the DFA shown in~\cref{fig:agent_operating} (on the right).}
    \label{fig:runningExMC}
\end{figure}

    We exemplify such semantics for MCs and for DFAs. As running examples, we 
    use (i) the MC presented in~\cref{fig:runningExMC}, which is a small variant of the MC in~\cref{fig:agent_operating} with the same transition probabilities, 
    and (ii) the DFA shown in~\cref{fig:runningExMC} ( and also in~\cref{fig:agent_operating}).

    \begin{example}[MC]
        \label{ex:tau-semanticsMC}
        We expect that the semantics from any state $s_1$ is the distribution $\sigma\in \distributionfunc(A^{+})$ over the traces that eventually reach the target state $\star$.
        For instance, the semantics $\sigma$ of the MC in~\cref{fig:runningExMC}, from the initial state, should be the distribution given by \[ \sand \cdot \lake \cdot \recharge \mapsto \nicefrac{4}{5},\quad \sand \cdot \sand \cdot \recharge \mapsto \nicefrac{4}{25}, \quad \sand \cdot \sand \cdot \volcano \mapsto \nicefrac{1}{25}.\]
        We show how to present this semantics as coalgebraic semantics.
        Let $c\colon X\rightarrow \mcfunc{X}$, and $c(x) = \big(P(x), l(x)\big)$ for $x\in X$.
        Formally, for each $x\in X$ and $w\in A^{+}$, the semantics $\sigma(x)(w)$ is inductively given by 
        \begin{align*}
            \sigma(x)(w) \defeq \begin{cases}
                    P(x, \star) &\text{ if $w = l(x)$},\\ 
                    \sum_{x'\in X}P(x, x') \cdot \sigma(x')(w') &\text{ if $w = l(x)\cdot w'$},\\
                    0 &\text{ otherwise. }
            \end{cases}
        \end{align*}
            
    This semantics can indeed be obtained as an instance of~\cref{def:sem_coal} with the predicate transformer $\Phi$.
        Intutively, the predicate transformer $\Phi$ provides a one-step update of reachability. 
        This idea of the one-step update is common in dynamic programming. In fact, the predicate transformer $\Phi$ for MCs is very similar to the Bellman operator, which has been widely studied in probabilistic model checking (see e.g.~\cite{Baier08}). 
    Consider the semantic structure $S = (\mathbf{\Omega}_S, \tau_S)$, where the semantic domain $\mathbf{\Omega}_S$ is given by $\Omega_S\defeq \fsublang$, ordered pointwise,\footnote{$\nu_1\preceq \nu_2$ if $\nu_1(w)\leq \nu_2(w)$ for any $w \in A^{+}$. We note that the least element exists in the set of subdistributions.} and the modality $\tau_S\colon \mcfunc{\fsublang}\rightarrow \fsublang $ is given by
        \begin{equation*}
        \tau_S(\nu, a)(w) \defeq \begin{cases}
            \nu(\star) &\text{ if } w = a,\\
            \sum_{\nu'\in \supp(\nu)\cap \fsublang}\nu(\nu')\cdot \nu'(w') &\text{ if } w = a\cdot w',\\
            0       &\text{ otherwise. }
        \end{cases}
        \end{equation*}
    This modality $\tau_S$ can be exemplified as follows. Consider two distribution over traces, \begin{align*}
& \nu_1\colon \lake \cdot \recharge\mapsto 1,
\qquad \nu_2\colon \sand\cdot\recharge \mapsto \nicefrac{4}{5}, \quad \sand\cdot\volcano\mapsto \nicefrac{1}{5}
 \end{align*}
 which correspond to the semantics of $x_1, x_2$, respectively. Now consider $\nu\colon \nu_1 \mapsto \nicefrac{4}{5}, \nu_2 \mapsto \nicefrac{1}{5}$, which can be intuitively mapped to state $x_0$. The modality helps us to construct the semantics for $x_0$ from $\nu$, i.e., from the semantics for the successors of $x_0$, via the induced predicate transformer  
    $\predtran{}$ on $\fsublang^X$.
    For instance, let $u\colon X\rightarrow \fsublang$ such that $u(x_1) = \nu_1$ and $u(x_2) = \nu_2$, and $w = \sand\cdot \lake \cdot \recharge$. 
    With the modality $\tau_S$, the predicate transformer $\predtran{}(u)(x_0)(w)$ is computed as follows:
    \begin{align*}
        \predtran{}(u)(x_0)(w) = \sum_{x'\in X} P(x_0, x')\cdot u(x')(w') = P(x_0, x_1)\cdot \nu_1(w') = 4/5, 
    \end{align*}
    where $w' \defeq \lake \cdot \recharge$. 
    Formally, the predicate transformer is given by
        \begin{equation*}
        \predtran{}(u)(x)(w) \defeq \begin{cases}
            P(x,\star) &\text{ if $w = l(x)$, }\\
            \sum_{x'\in X} P(x, x')\cdot u(x',w') &\text{ if  $w = l(x)\cdot w'$,}\\
            0 &\text{ otherwise, }
        \end{cases}
        \end{equation*}
    for each $u\colon X\rightarrow \fsublang$, $x\in X$ and $w\in A^{+}$.
    Then, for each $x\in X$, the coalgebraic semantics $\lfpoint[x]{c}$ shown in~\cref{fig:coalgebraicTSwithfibration} is exactly the semantics $\sigma(x)$ defined concretely above. 
        
    \end{example}

    \begin{example}[DFA]
        \label{ex:tau-semanticsDFA}
        We expect that the semantics of a state of a DFA is its recognized language $L\in \nempfaclang$ excluding the empty string $\epsilon$. 
    Concretely, the semantics of the DFA shown in~\cref{fig:agent_operating} represents the requirement for the robot, discussed before. 
    For example, the semantics includes the trace $w_2 = \sand\cdot \sand \cdot \recharge$, 
    but does not include the trace $w_1 =  \sand \cdot \lake\cdot \recharge$. 
    We formulate this semantics for (coalgebraic) DFAs. 
    Let $d\colon Y\rightarrow \dfautomata{Y}$ be a DFA.
    Formally, for each $y\in Y$, the semantics $L(y)\colon A^{+}\rightarrow \settorf (\cong \nempfaclang)$ is inductively given by 
        \begin{align*}
            L(y)(w) \defeq \begin{cases}
                \top &\text{ if $w = a$ and $d(y,a) = (\_, \top)$},\\
                 &\text{ or $w = a\cdot w'$, $d(y,a) = (y', \_)$, and  $L(y')(w') = \top$, }\\
                 \bot &\text{ otherwise.}
            \end{cases}
        \end{align*}
        
    This is recovered through the semantic structure $R=(\mathbf{\Omega}_R, \tau_R)$, where the semantic domain $\mathbf{\Omega}_R$ is given by $\Omega_R\defeq \nempfaclang$, and $T_1\preceq T_2$ if $T_1\subseteq T_2$ and the modality $\tau_R\colon \dfautomata{\nempfaclang}\rightarrow \nempfaclang $,
    \begin{equation*}
        \tau_R(\delta)\defeq \{ a \mid a\in A,\, \delta(a) = (\_, \top) \} \cup \{a\cdot w \mid a\in A,\, \delta(a) = (T, \_),\, w\in T\}.
    \end{equation*}
	The first part says that a single letter word $a$ is accepted if the transition output is $\top$, and a word of the form $a \cdot w \in A \cdot A^+$ is
	accepted if, after an $a$-transition, the word $w$ is accepted.
	
    The induced predicate transformer $\predtran{}$ on $(\settorf^{A^{+}})^Y$  provides a one-step update of recognized languages. 
    Given a state and a subset of the recognized language of each successor state, the predicate transformer updates the subset of the current state by these successor states in the backward manner.
    Formally, the induced predicate transformer $\predtran{}$ on $(\settorf^{A^{+}})^Y$ 
    is given by 
    \begin{equation*}
        \predtran{}(u)(y)(w) \defeq \begin{cases}
                \top &\text{ if $w = a$ and $d(y)(a) = (\_, \top)$},\\
                 &\text{ or $w = a\cdot w'$, $d(y)(a) = (y', \_)$, and  $u(y')(w') = \top$, }\\
                 \bot &\text{ otherwise,}
            \end{cases}
    \end{equation*}
    for each $u\colon Y\rightarrow 2^{A^+}$, $y\in Y$ and $w\in A^{+}$.
    Again, for each $y\in Y$, the semantics $\lfpoint[y]{d}$ coincides with the recognized language $L(y)$. 
    
    \end{example}

\subsection{Coalgebraic Inference}\label{ssec:coalg-inference-queries}

The semantics of coalgebras can be used to define the observable behaviour of both the system and the requirements, coalgebras that typically have different functors. 
To formulate the inference problem, we need a last ingredient: \emph{queries}.
Queries define the semantics of the system and the requirement \emph{combined}. Formally, a query is a  function $q\colon \tvdsys{\Omega} \times \tvdspec{\Omega} \rightarrow \Omega$, taking semantic values for system and requirement respectively, and returning a semantic value in a third domain $\Omega$. If we infer a probability, then we may expect $\Omega = [0,1]$.

\begin{definition}[coalgebraic inference]\label{def:coalg-inference}
	Consider (i)~a system  $c \colon X \rightarrow \fsys{F}(X)$ with semantic domain $\tvdsys{\Omega}$, (ii)~a requirement $d\colon Y \rightarrow \fspec{F}(Y)$ with semantic domain $\tvdspec{\Omega}$, and (iii)~a query $q\colon \tvdsys{\Omega} \times \tvdspec{\Omega} \rightarrow \Omega$.
	The \emph{coalgebraic inference map} is the composite
	\[
	q \circ\big(\lfpoint{c}\times \lfpoint{d}\big)\colon X\times Y\rightarrow \Omega \,.
	\]
\end{definition}
We refer to the computation of this function as \emph{(coalgebraic) inference}. In particular, fixing an initial state $y \in Y$ of the requirement, the type of the function $q \circ\big( \lfpoint{c}\times \lfpoint[y]{d}\big)\colon X \rightarrow \Omega$ matches the notion of expectations as is common in verification calculi for probabilistic programs~\cite{DBLP:journals/jcss/Kozen85,DBLP:series/mcs/McIverM05,DBLP:phd/dnb/Kaminski19}.

\begin{example}[probabilistic inference]
\label{ex:probinf_coalinf}
For the MC and the DFA presented in~\cref{ex:tau-semanticsMC,ex:tau-semanticsDFA}, the only accepting trace is $w_2 = \sand \cdot \sand \cdot \recharge$. 
Consequently, we expect that the result of probabilistic inference for the MC and the DFA is $\sigma(w_2)\defeq \frac{1}{5}\cdot \frac{4}{5}$. 
We see that this probabilistic inference is an instance of the coalgebraic inference defined in~\cref{def:coalg-inference}.
We have (i)~system $c \colon X\rightarrow \mcfunc{X}$ with $\tvdsys{\Omega} = \fsublang$, and (ii)~requirement $d \colon Y\rightarrow \dfautomata{Y}$ with $\tvdspec{\Omega} = \nempfaclang$ as before. 
For the inference problem from \cref{sec:agent_operation}, $\Omega = [0,1]$ and  we define the query as:
\begin{align*}
    q\colon \fsublang\times\nempfaclang\rightarrow \probinterval\,, \qquad q(\nu, L) \defeq \sum_{w\in L} \nu(w).
\end{align*}
Then, for each state $y\in Y$, coalgebraic inference as given in \cref{def:coalg-inference} is the map
\[q\big( \lfpoint{c},\,\lfpoint[y]{d}\big)\colon X\rightarrow \probinterval, \quad q(\lfpoint[x]{c}, \lfpoint[y]{d}) \defeq \sum_{w\in \lfpoint[y]{d}}\lfpoint[x]{c}(w).\]

\end{example}

\section{Product Constructions for Coalgebraic Inference}
\label{sec:coal_model_checking_regular_semantics}
 In this section, we introduce a \emph{coalgebraic product construction}.
Critically, this allows us to reduce the problem of coalgebraic inference in~\cref{ssec:coalg-inference-queries} to that of computing the semantics on the product.

\subsection{Coalgebraic Product}\label{ssec:coalgProd}

To define coalgebraic product constructions formally, we assume a functor $\fsys{F}$ representing the type of systems, a functor $\fspec{F}$ for requirements, and a third functor $\fprod{F}$ which represents the type of coalgebras resulting from a product construction between systems and requirements. 
Examples of this setup are summarized in \cref{fig:list_of_coalgebraic_inference}, with pointers to sections where they are discussed in detail.

\begin{figure}
    \scalebox{0.7}{
    \begin{tabular}{lcccccc}\toprule
    & \multicolumn{2}{c}{System} &\multicolumn{2}{c}{Requirement}  &\multicolumn{2}{c}{Product} 
    \\\cmidrule(lr){2-3}\cmidrule(lr){4-5}\cmidrule(lr){6-7}
    Section& $\fsys{F}$ & $\tvdsys{\Omega}$ & $\fspec{F}$ & $\tvdspec{\Omega}$ & $\fprod{F}$ & $\tvdprod{\Omega}$ \\\midrule
    \cref{sec:coalgebraic_inference},\cref{sec:coal_model_checking_regular_semantics},\cref{subsec:cost-boundedReach} & $\mcfunc{X}$ & $\fsublang$ & $\dfautomata{Y}$ & $\nempfaclang$ &  $\mcflagfunc{X\times Y}$ & $\probinterval$ \\
    \cref{sec:rewards_prob_program} & $\mcfunc[\nat\times A]{X}$ & $\subdistributionfunc(A^{+}\times \nat)$ & $\dfautomata{Y}$ & $\nempfaclang$ &  $\distributionfunc(X\times Y+\settorf) \times \nat$ & $\probinterval\times \nonegreal$ \\
    \cref{subsect:costReq}  & $\mcfunc{X}$ & $\fsublang$ & $\dfautomata[A]{Y\times \nsetbot{N}}$ & $\nempfaclang$ &  $\mcflagfunc{X\times Y\times \nsetbot{N}}$ & $\probinterval$ \\
    \cref{sec:MCsNeverTerminate} & $ \mcinffunc{X}$ & $\probmeasure(\salg{A})\cup \bigcup_{n\in \nat} \distributionfunc(A^n)$ & $\dfautomata{Y}$ & $\nempfaclang$ &  $\dmcflagfunc{X\times Y }$ & $\probinterval$ \\
    \cref{sec:weighted_programming} & $\ltsfuncsemiring{X}{\nat}$ & $\powersetfunc(A^{+}\times \nat)$ & $\exepnfautomata{Y}$ & $\nempfaclang$ &  $\suptlflagfunc{X\times Y}{\nat}$ & $\natinf$ \\
    \cref{sec:quantitativeInferenceWeightedSystemRequirement}    & $\ltsfuncsemiring{X}{\nat}$ & $\powersetfunc(A^{+}\times \nat)$ & $\wmealyfunc{Y}$ & $\powersetfunc(A^{+}\times \nat)$ &  $\suptlflagfunc{X\times Y}{\nat}$ & $\natinf$ \\   
    \bottomrule
    \end{tabular}
    }
    \caption{Coalgebraic inference and product. We write $\probmeasure(\salg{A})$ for the set of probability measures on the smallest $\sigma$-algebra that contains all cylinder sets; see~\cref{def:sigma_algebra} for the definition. 
    }
    \label{fig:list_of_coalgebraic_inference}
\end{figure}

\begin{figure}
    \begin{minipage}[b]{0.45\textwidth}
        \begin{tikzcd}
            \fsys{F}(X)\times \fspec{F}(Y)  \arrow[d, "\lambda_{X, Y}"] \arrow[rr, "\fsys{F}(f)\times \fspec{F}(g)"] & & \fsys{F}(U)\times \fspec{F}(V) \arrow[d, "\lambda_{U, V}"]  \\
            \fprod{F}(X\times Y) \arrow[rr, "\fprod{F}(f\times g)"]& & \fprod{F}(U\times V)
            \end{tikzcd}
            \subcaption{The naturality of $\lambda$}
    \end{minipage}
    \hfill
        \begin{minipage}[b]{0.45\textwidth}
            \vspace{-1cm}
        \scalebox{0.8}{
            \begin{tikzpicture}
                \node[state] (s0) at (-1.5, 0) {$x_0, y_0$};
                \node[state] (s1) at (0, 0.75) {$x_1, y_0$};
                \node[state] (s2) at (0, -0.75) {$x_2, y_0$};
                \node[state] (s3) at (1.5, 0.75) {$x_3, y_2$};
                \node[state] (s4) at (1.5, -0.75) {$x_3, y_0$};
                \node[state] (s5) at (3, 0.75) {$\bot$};
                \node[state] (s6) at (3, -0.75) {$\top$};
                \node[inner sep=0,rotate=150] at (0.65cm, -1.2cm) (cdots1) {$\cdots$};
                \draw[->] (s0) edge[above] node{$\frac{4}{5}$} (s1);
                \draw[->] (s1) edge[above] node{$1$} (s3);
                \draw[->] (s0) edge[above] node{$\frac{1}{5}$} (s2);
                \draw[->] (s2) edge[above] node{$\frac{4}{5}$} (s4);
                \draw[->] (s3) edge[above] node{$1$} (s5);
                \draw[->] (s4) edge[above] node{$1$} (s6);
            \end{tikzpicture}  
            }
            \subcaption{A fragment of the product}
        \end{minipage}
    \caption{Illustrations of (a) the naturality of a distributive law $\lambda$, and (b) a fragment of the product of the MC and the DFA shown in~\cref{fig:runningExMC}}
    \label{fig:productExMC}
\end{figure}

\begin{definition}[coalgebraic product]
    Let $\fsys{F}, \fspec{F}$ and $\fprod{F}$ be functors. Let $c\colon X\rightarrow  \fsys{F}(X)$ and $d\colon Y\rightarrow \fspec{F}(Y)$  be coalgebras, and assume a map $\lambda_{X,Y} \colon  \fsys{F}(X) \times \fspec{F}(Y) \rightarrow \fprod{F}(X \times Y)$. 
    The \emph{coalgebraic product} $c\otimes_{\lambda} d$ is the coalgebra $c\otimes_{\lambda} d\colon X\times Y\rightarrow \fprod{F}(X\times Y)$,  given by $c\otimes_{\lambda} d\defeq \lambda_{X, Y} \circ (c\times d)$:
    \begin{equation*}
\begin{tikzcd}
       X\times Y \arrow[r,"c\times d"] & \fsys{F}(X)\times \fspec{F}(Y) \arrow[r,"\lambda_{X, Y}"] & \fprod{F}(X\times Y).
    \end{tikzcd}
\end{equation*}
\end{definition}
The essence of the product construction is captured by the map $\lambda_{X,Y}$, which explains how to move from ``behaviours'' in $\fsys{F}$ and $\fspec{F}$ to a behaviour in the product type $F_{S \otimes R}$. It is independent of the coalgebras at hand, and in fact it should be defined uniformly for all sets $X$ and $Y$, and consequently for all systems and requirements (with the same functors). We capture this uniformity by assuming that $\lambda_{X,Y}$ extends to a \emph{natural transformation}, which we refer to as a distributive law:

\begin{definition}[distributive law $\lambda$]
\label{def:distributive_law}
Let $\fsys{F},\fspec{F}, \fprod{F}$ be functors. A \emph{distributive law} $\lambda$ 
 from $\fsys{F}$ and $\fspec{F}$ to $\fprod{F}$ is a collection of maps $\lambda_{X,Y} \colon \fsys{F}(X) \times \fspec{F}(Y) \rightarrow \fprod{F} (X \times Y)$, one for each pair of sets $X$ and $Y$, which is \emph{natural} in $X$ and $Y$. The \emph{naturality of $\lambda$} means that for any two functions $f\colon X\rightarrow U$ and $g\colon Y\rightarrow V$ the diagram that is shown in~\cref{fig:productExMC} commutes.
\end{definition}

In the construction of coalgebraic products, we do not use the naturality of distributive laws. However, the naturality of distributive laws is essential for the proof of the main theorem (\cref{thm:correctness_product}), which ensures the ``correctness'' of the coalgebraic product construction.

\begin{example}[product of MCs and DFAs]
    \label{ex:dis_mc_dfa}
Consider the running examples shown in~\cref{ex:tau-semanticsMC} and~\cref{ex:tau-semanticsDFA}.  
A fragment of the traditional product~\cite{Baier08} of the MC and the DFA is illustrated in~\cref{fig:productExMC}.
Its states are pairs of a state in the original MC and a state in the DFA, the probability transition is the synchronization of the MC and the DFA by reading the labels on the MC.
There are two sink states $\bot$ and $\top$ in the product, 
and the reachability probability to the sink state $\top$ coincides with the probabilistic inference of the MC and the DFA.  

We formally recall the traditional product of MCs and DFAs.  
Let $c\colon X\rightarrow \mcfunc{X}$ be an MC, and $d\colon Y\rightarrow \dfautomata{Y}$ be a DFA. The \emph{product} of $c$ and $d$ is the MC $f\colon X\times Y \rightarrow \mcflagfunc{X\times Y}$ without labels defined as follows: for each $(x_1, y_1)\in X\times Y$, $f(x_1, y_1) \defeq \nu$ such that 
\begin{align*}
    \nu(x_2, y_2) &\defeq \begin{cases}
        P(x_1, x_2) &\text{ if $c(x_1) = \big(P(x_1,\_), l(x_1)\big)$, and  $d(y_1)(l(x_1)) = (y_2, \_)$},\\
        0 &\text{ otherwise,}
    \end{cases}\\
    \nu(b) &\defeq \begin{cases}
        P(x_1, \star) &\text{ if $c(x_1) = \big(P(x_1,\_), l(x_1)\big)$, and  $d(y_1)(l(x_1)) = (\_, b)$},\\
        0 &\text{ otherwise,}
    \end{cases}
\end{align*}
for each $(x_2, y_2)\in X\times Y$. 
This product $f$ is precisely the coalgebraic product $c \otimes_{\lambda} d$ with the distributive law $\lambda$ from the functors $F_S\defeq \mcfunc{\_}$ and $F_R\defeq \dfautomata{\_}$ to $F_{S\otimes R}\defeq \mcflagfunc{\_}$  below. 
    The \emph{distributive law} $\lambda_{X, Y}\colon \mcfunc{X}\times \dfautomata{Y} \rightarrow \mcflagfunc{X\times Y}$ for MCs and DFAs is given by
\begin{align*}
    \lambda_{X, Y}(\nu, a, \delta)(x, y) &\defeq \begin{cases}
        \nu(x) &\text{ if }\delta(a) = (y, \_),\\
        0 &\text{ otherwise,}
    \end{cases}
    &\lambda_{X, Y}(\nu, a, \delta)(b) \defeq \begin{cases}
        \nu(\star) &\text{ if }\delta(a) = (\_, b),\\
        0 &\text{ otherwise.}
    \end{cases}
\end{align*}
\end{example}

\begin{remark}[map that is not natural]
\label{rem:non_natural}
We show a map $\lambda_{X,Y} \colon \subdistributionfunc(X)\times \powersetfunc(Y)\rightarrow \subdistributionfunc(X \times Y)$ that is not natural:
for each $\nu\in \subdistributionfunc(X)$, $S\in \powersetfunc(Y)$, and $(x, y)\in X\times Y$, the map $\lambda_{X,Y}(\nu, S)(x, y)$ is defined by 
$\lambda(\nu, S)(x, y) \defeq 1/|S|\cdot \nu(x) $ if $S\not = \emptyset$ and $y\in S$, and $0$ otherwise.
Given this fact, we conjecture that finding product MCs of MCs and NFAs for the quantitative temporal inference (without the determinisation of NFAs) is challenging (probably impossible).

We also note that there are product constructions for MCs and \emph{separating automata}~\cite{CouvreurSS03} or \emph{unambiguous automata}~\cite{BaierK00023} without determinization.
The product construction in~\cite{BaierK00023}, which involves three different types of automata, can be formulated in our framework. 
The standard semantics of for such automata, however, is not the sought-after solution on our product without making the additional assumptions regarding ambiguity and we would leave this as future work.
\end{remark}
\subsection{Correctness of Product Constructions}

The purpose of the product construction is to answer inference queries for a system $c$ and a requirement $d$ by computing semantics on the product $c \otimes_{\lambda} d$.
We formulate the \emph{correctness} of this approach, i.e., that  computing the semantics on the product indeed solves the actual inference problem. 
For instance, the reachability probability to the state $\top$ on the product of an MC and a DFA coincides with the probabilistic inference $q\big( \lfpoint{c},\,\lfpoint{d}\big)$ of the MC and the DFA defined in~\cref{ex:probinf_coalinf}; 
this is the correctness of the product for MCs and DFAs w.r.t. the query $q$.

The key observation here is that the coincidence of the semantics of products and inferences can be captured by an equality between semantics of products and inferences: we formulate this property as \emph{correctness} of the product construction.

\begin{definition}[correctness]\label{def:correctness}
Assume the following: 
\begin{itemize}
    \item  a (system) functor $F_S$ and semantic structure $(\tvdsys{\mathbf{\Omega}}, \modsys{\tau})$, a (requirement) functor $F_R$ and semantic structure $(\tvdspec{\mathbf{\Omega}}, \modspec{\tau})$, and a query $q\colon\tvdsys{\Omega}\times \tvdspec{\Omega}\rightarrow \Omega_{S\otimes R}$ (as in \cref{def:coalg-inference}),
    \item a (coalgebraic product) functor $F_{S\otimes R}$ and semantic structure $(\mathbf{\Omega}_{S \otimes R}, \tau_{S \otimes R})$,
    \item a distributive law $\lambda$ from $\fsys{F}$ and $\fspec{F}$ to $F_{S\otimes R}$,
\end{itemize}
The product $c \otimes_{\lambda} d$ with the modality $\tau_{S\otimes R}$ is \emph{correct w.r.t.\ the query $q$} if 
\begin{equation}
        \lfpoint{c\otimes_{\lambda} d} \quad=\quad q\circ (\lfpoint{c} \times \lfpoint{d})
         \label{eq:def_correctness}
    \end{equation}
for all coalgebras $c \colon X\rightarrow F_S(X)$ and $d \colon Y \rightarrow F_R(Y)$.
\end{definition}

\begin{example}[semantic structure for the product of MCs and DFAs]
    \label{ex:ssProdMCDFA}
    Recall the distributive law $\lambda_{X, Y}\colon \mcfunc{X}\times \dfautomata{Y}\rightarrow \mcflagfunc{X\times Y}$ given in~\cref{ex:dis_mc_dfa} for the coalgebraic product of MCs and DFAs. 
    We define the semantic structure $\big((\probinterval, \leq), \tau_{S\otimes R}\big)$ for the functor $F_{S \otimes R}\defeq \mcflagfunc{-}$ as follows. 
    The modality $\tau_{S\otimes R}\colon \mcflagfunc{\probinterval}\rightarrow \probinterval$ is given by 
    \begin{align*}
        \tau_{S\otimes R}(\nu) &\defeq \nu(\top) + \sum_{p\in \probinterval} \nu(p)\cdot p, \text{ where the support $\supp(\nu)$ is at most countable. }
    \end{align*}
    Similar to \cref{ex:tau-semanticsMC}, the informal intuition of the modality $\tau_{S\otimes R}$ is the following: the reachability probability from the current state to the state $\top$ is the sum of (i)~the transiton probability to the state $\top$, 
    and (ii)~a weighted sum of  reachability probabilities from the successor states $(p)$, weighted by the transition probabilities $\nu(p)$ to the successor states. Note that the successors are represented by their reachability probability $p$.
    
    The predicate transformer $\predtran{}$ on $\probinterval^{X\times Y}$ is now given by 
    \begin{align*}
        \predtran{}(u)(x, y) & \defeq \lambda_{X, Y}\big(c(x), d(y)\big)(\top) + \sum_{(x',y')\in X\times Y} \lambda_{X, Y}\big(c(x), d(y)\big)(x', y')\cdot u(x', y')\\
        &= P(x, \star) \cdot \dirac{b_{y, l(x)} = \top} + \sum_{x'\in X} P(x, x') \cdot u(x', n_{y, l(x)}),
    \end{align*}
        where $\big(P(x), l(x)\big) \defeq c(x)$, and $(n_{y, l(x)},\, b_{y, l(x)}) \defeq d(y)\big(l(x)\big)$. 
    Here, the semantics $\lfpoint{c\otimes_{\lambda} d}$ induced by the predicate transformer $\Phi$ coincides with the inference.
\end{example}

As we illustrated in~\cref{sec:overview} for a range of examples, computing $\lfpoint{c\otimes_{\lambda} d}$ is a practical solution for inference $ q\circ( \lfpoint{c} \times \lfpoint{d})$ if Eq.~\cref{eq:def_correctness} holds. 
We now turn to a sufficient condition for correctness. 
In fact, there is a simple but powerful correctness criterion---the main theorem of the paper---that roughly is given as follows: a query $q$ (i) preserves the $\omega$-cpo structures, and (ii) makes a diagram, given below, commute. We observe that the diagram involves a distributive law $\lambda$ and the modalities. In the proof, the naturality of $\lambda$ is essential.

\begin{theorem}[correctness criterion]
\label{thm:correctness_product}
Consider the data from \cref{def:correctness}, and assume the following:
\begin{itemize}
    \item $q$ preserves the least elements, i.e.,\ $q(\bot, \bot) = \bot$,
    \item $q$ is $\omega$-continuous, i.e.,\ for all $\omega$-chains $(t_i)_{i\in \nat}$ and $(t'_j)_{j\in \nat}$, the following equality holds:
    \begin{equation*}
        q\big(\bigvee_{i\in \nat} t_i,\, \bigvee_{j\in \nat} t'_j\big) = \bigvee_{k\in \nat} q(t_k, t'_k),
    \end{equation*}
    \item the following diagram commutes:
    \begin{equation}
    \label{eq:coherence_evaluation}
    \begin{tikzcd}
            \fsys{F}(\tvdsys{\Omega})\times \fspec{F}(\tvdspec{\Omega})\arrow[rr,"\modsys{\tau}\times \modspec{\tau}"] \arrow[d,"\lambda_{\tvdsys{\Omega}, \tvdspec{\Omega}}"]
            & 
            &\tvdsys{\Omega}\times \tvdspec{\Omega} \arrow[d,"q"]\\
F_{S\otimes R}(\tvdsys{\Omega}\times\tvdspec{\Omega}) \arrow[r,"F_{S\otimes R}(q)"]
            & F_{S\otimes R}(\Omega_{S\otimes R}) \arrow[r, "\tau_{S\otimes R}"]
            & \Omega_{S\otimes R}
    \end{tikzcd}
    \end{equation}
\end{itemize}
Then, the product construction with $\tau_{S\otimes R}$ is correct w.r.t.\ $q$ (in the sense of \cref{def:correctness}). 
\end{theorem}
In addition to the data that are needed to define the problem of coalgebraic inference, 
the correctness criterion requires us to specify (i) a product functor $F_{S\otimes R}$; (ii)  a distributive law $\lambda$; and (iii) a modality $\tau_{S\otimes R}$,
and prove that the associated diagram commutes. 
These ingredients are precisely what is needed for our framework to form the basis of efficient inference.

\begin{example}[correctness for probabilistic inference]
We show that the query $q\colon\fsublang\times \nempfaclang\rightarrow \probinterval$ from \cref{ex:probinf_coalinf} satisfies the correctness criterion. Recall that 
\begin{equation*}
    q(\nu, L) \defeq \sum_{w\in L}\nu(w).
\end{equation*}
It is straightforward to prove that $q$ is $\omega$-continuous by the monotone convergence theorem 
\iffull
(see~\cref{prop:omega_cpo_e_mc_dfas} for details). 
\else
(see~\cite[Appendix A]{Watanabe24} for details). 
\fi
The diagram shown in Eq.~\cref{eq:coherence_evaluation} for the query $q$ is the following: 
\begin{equation*}
\begin{tikzcd}
            \mcfunc{\fsublang}\times \dfautomata{\nempfaclang}\arrow[rr,"\modsys{\tau}\times \modspec{\tau}"] \arrow[d,"\lambda_{\fsublang, \nempfaclang}"]
            & 
            &\fsublang\times \nempfaclang \arrow[d,"q"]\\
\mcflagfunc{\fsublang\times \nempfaclang} \arrow[r,"\distributionfunc(q + \settorf)"]
            & \mcflagfunc{\probinterval} \arrow[r, "\modprod{\tau}"]
            & \probinterval
\end{tikzcd}
\end{equation*}
where for $\nu\in \mcflagfunc{\fsublang\times \nempfaclang}$, the value $\distributionfunc(q + \settorf)(\nu)$ can be concretely described by 
\begin{align*}
    &\distributionfunc(q + \settorf)(\nu)(p)\defeq \sum_{q(\nu', L) = p} \nu(\nu', L)\, \text{ (for $p\in \probinterval$)}, & \distributionfunc(q + \settorf)(\nu)(b)\defeq \nu(b)\, \text{ (for $b\in \settorf$).}
\end{align*}
A direct calculation shows that the diagram actually commutes 
\iffull 
(see~\cref{lem:commute_mc_dfas}),
\else
(see~\cite[Appendix A]{Watanabe24}), 
\fi
thus the coalgebraic product construction given by the distributive law~\cref{ex:dis_mc_dfa} is correct. 

\end{example}

It is often convenient to restrict Eq.~\cref{eq:coherence_evaluation} to the Kleene iterations induced by predicate transformers. 
This gives a weaker correctness criterion than~\cref{thm:correctness_product}, but is still sufficient to entail the correctness of product constructions. 
\begin{proposition}
    \label{prop:weaker_correctness_criterion}
    Consider the data from \cref{def:correctness}, and assume the following:
    \begin{itemize}
       \item $q$ preserves the least elements and is $\omega$-continuous (as in~\cref{thm:correctness_product}),
       \item  for any coalgebras $c \colon X\rightarrow F_S(X)$ and $d \colon Y \rightarrow F_R(Y)$, and $k\in \nat$ , 
       the following diagram commutes with $u_{S, k} \defeq (\predtran{\modsys{\tau}, c})^{k}(\bot)$ and $u_{R, k} \defeq (\predtran{\modspec{\tau}, d})^{k}(\bot)$:
       \begin{equation}
       \label{eq:coherence_evaluation_2}
       \begin{tikzcd}
        F_S(X)\times F_R(Y)\arrow[rr,"F_S(u_{S, k})\times F_R(u_{R, k})"]\arrow[d,"F_S(u_{S, k})\times F_R(u_{R, k})"]&&\fsys{F}(\tvdsys{\Omega})\times \fspec{F}(\tvdspec{\Omega})\arrow[rr,"\modsys{\tau}\times \modspec{\tau}"] 
               & 
               &\tvdsys{\Omega}\times \tvdspec{\Omega} \arrow[d,"q"]\\
   F_S(\tvdsys{\Omega})\times F_R(\tvdspec{\Omega})\arrow[rr,"\lambda_{\tvdsys{\Omega}, \tvdspec{\Omega}}"]&&F_{S\otimes R}(\tvdsys{\Omega}\times\tvdspec{\Omega}) \arrow[r,"F_{S\otimes R}(q)"]
               & F_{S\otimes R}(\Omega_{S\otimes R}) \arrow[r, "\tau_{S\otimes R}"]
               & \Omega_{S\otimes R}
       \end{tikzcd}
       \end{equation}
    \end{itemize}
    Then, the product construction  with $\tau_{S\otimes R}$ is correct w.r.t.\ $q$.
\end{proposition}
The proof is the same as~\cref{thm:correctness_product}.
We apply~\cref{prop:weaker_correctness_criterion} only for the probabilistic temporal inference for programs that never terminate in~\cref{sec:MCsNeverTerminate}.

\section{Case Study \Rom{1}: Partial Expected Rewards in Probabilistic Programs}
\label{sec:rewards_prob_program}

We show that the coalgebraic product construction also works for \emph{Markov reward models}, where the inference problem is to compute \emph{partial expected rewards}~\cite{Baier08,Baier0KW17}  over the accepting traces determined by a DFA.  
This extension is useful, for example, for modelling of the computation time of randomized algorithms~\cite{DBLP:phd/dnb/Kaminski19,KaminskiKMO18}. 


 A \emph{Markov reward model} is a coalgebra $c\colon X\rightarrow \mcfunc[\nat\times A]{X}$. 

The semantic domain for Markov reward models is given by $\big(\subdistributionfunc(A^{+}\times \nat) , \preceq\big)$, where $\preceq$ is defined pointwise.
The semantics of a coalgebra $c$ thus is a map $\lfpoint{c} \colon X \rightarrow \subdistributionfunc(A^{+}\times \nat)$, where the underlying set $A^{+}\times \nat$ represents a pair of a trace $w\in A^{+}$, 
and one of the possible accumulated rewards $r\in \nat$ along the corresponding trace $w$.
We note that in general there are multiple different paths whose traces are the same and whose accumulated rewards are different.
An informal description of the partial expected reward is to collect paths that reach the target state $\star$, and take the sum of the multiplications of the path probability of each such path by its accumulated sum of rewards. 
Importantly, unlike the standard expected reward, the partial expected reward is finite even if the probability of reaching the target state is strictly less than $1$. 
See~\cite{Baier08,Baier0KW17}, for instance, for the definition of the partial expected reward.
The semantics is equivalently defined in the appendix (\cref{subsec:MRM}) in terms of a semantic structure.

We assume that the requirement for the Markov reward machine is given by the DFA $d\colon Y\rightarrow \dfautomata{Y}$ defined in~\cref{ex:mc_moore_mealy}. 
We then define the query $q \colon \subdistributionfunc(A^{+}\times \nat)\times \nempfaclang\rightarrow \probinterval\times \nonegreal$ for Markov reward models, where the semantic domain $\probinterval\times \nonegreal$ is a pair of the reachability probability and the partial expected reward to the target state; we have to compute reachability probabilities since we do not assume reachability probabilities are always $1$. 
\begin{definition}
\label{def:query_rewards}
The \emph{query $q:\subdistributionfunc(A^{+}\times \nat)\times \nempfaclang\rightarrow \probinterval\times \nonegreal $ for Markov reward models} is given by $q(\nu, L) \defeq (p, r)$, where 
    \begin{align*}
        p &\defeq \sum_{n\in \nat}\nu\Big(\big\{(w, n) \mid w\in L \big\}\Big)\,,
        &r \defeq \sum_{n\in \nat} n\cdot \nu\Big(\big\{(w, n) \mid w\in L \big\}\Big).
    \end{align*}
\end{definition}

We define the distributive law for Markov reward models with in a similar way as given  in~\cref{ex:dis_mc_dfa}. 

\begin{definition}
\label{def:dist_rewards}
    The \emph{distributive law $\lambda$ for Markov reward models} is given by $\lambda_{X, Y}: \mcfunc[\nat\times A]{X}\times \dfautomata{Y} \rightarrow \distributionfunc(X\times Y+\settorf) \times \nat$, where $\lambda_{X, Y}(\nu_1, m, a, \delta)\defeq (\nu_2, m)$ such that 
    \begin{align*}
    \nu_2(x, y) &\defeq \begin{cases}
        \nu_1(x) &\text{ if }\delta(a) = (y, \_),\\
        0 &\text{ otherwise,}
    \end{cases}
    &\nu_2(b) \defeq \begin{cases}
        \nu_1(\star) &\text{ if }\delta(a) = (\_, b),\\
        0 &\text{ otherwise.}
    \end{cases}
\end{align*}
\end{definition}

The product with the modality $\tau_{S\otimes R}$ for the partial expected reward that is given in~\cref{subsec:MRM} satisfies the correctness criterion (\cref{thm:correctness_product}).  See~\cref{subsec:MRM} for the proof.
\begin{proposition}
 The product induced by $\lambda$ defined in ~\cref{def:dist_rewards} with the modality $\tau_{S\otimes R}$  is correct with respect to the query $q$. 
\end{proposition}

\begin{remark}
As the query in \cref{def:query_rewards} computes the probability \emph{and} the partial expected reward, it is straightforward to use the construction also for conditional rewards (see \cref{ssec:ext-inference}).
\end{remark}

\section{Case Study \Rom{2}: Resource-Sensitive Reachability of MCs}
\label{sec:cost_bounded_consumption}
In this section, we present resource-sensitive reachability analyses of MCs. 
Firstly, we show that reachability analysis to the target state for resource-sensitive MCs---\emph{MCs with costs}~\cite{SteinmetzHB16,HartmannsJKQ20}---is an instance of our coalgebraic inference (in~\cref{subsec:cost-boundedReach}). 
We then introduce a resource-sensitive reachability analysis for MCs, 
where resources are induced by requirements (in~\cref{subsect:costReq}); the resource-sensitive inference is again an instance of coalgebraic inference. 
We show that all of these instances meet the correctness criterion (\cref{thm:correctness_product}), implying their product constructions are correct (by~\cref{prop:correctCostBound,prop:correctCostReq}).

\subsection{Cost-Bounded Reachability}
\label{subsec:cost-boundedReach}

MCs with costs assign to every state a cost of entering that state, or alternatively, to every transition a cost of taking that transition.
Given $M \in \nat$, an \emph{MC with (bounded) costs and a target state} is a coalgebra $c\colon X\rightarrow \mcfunc[\nset{M}]{X}$. 
The requirement for cost-bounded reachability \sj{why thus? Did we even clarify what cost-bounded reach is? Shouldnt we introduce the coalgebra for the requirement?} tracks the sequences of the costs, where the specified bound for the accumulated sum of the sequences is given by a natural number $N$.
Concretely, let $\nu\in \subdistributionfunc(\nset{M}^{+})$ be the subdistribution of the sequence of costs, where the paths eventually reach the target state $\star$. 
The cost-bounded reachability is the inference \[\cbquery\colon \subdistributionfunc(\nset{M}^{+}) \times \nat \rightarrow \probinterval\quad\text{ with }\quad\cbquery(\nu, N)\defeq \sum_{w\in \supp(\nu), \sum w  < N} \nu(w),\] where $\sum w$ denotes the accumulated sum of the costs on $w$. 

We formulate the inference as the instance of our coalgebraic inference (\cref{def:coalg-inference}). 
The requirement is now given by the following DFA: 
\sj{I dont understand why we formulate this as a DFA. Why not a as a coalgebra?}
\begin{definition}[requirement for MCs with costs]
    \label{def:dfas_for_costs}
     Let $M,N \in \nat$. The \emph{requirement} $d$ is the coalgebra $d\colon\nsetbot{N}\rightarrow \dfautomata[\nset{M}]{\nsetbot{N}}$ given, 
    for each $i\in \nset{N}$ and $j\in \nset{M}$, by
    \begin{align*}
        d(i)(j) &\defeq \begin{cases}
            \big(i - j, \top \big) &\text{ if } i - j > 0,\\
            \big(\bot, \bot\big) &\text{ otherwise,} 
        \end{cases}\,
        &d(\bot)(j) \defeq (\bot, \bot).
    \end{align*}    
\end{definition}
Given a state $m\in \nset{N}$ of the requirement $d$, the semantics $\lfpoint[m]{d}$ defined in~\cref{ex:tau-semanticsDFA} is the recognized language of the requirement $d$ from the state $m$, that is, 
the set $\{ w \mid w\in \nset{M}^{+} \text{ such that }\sum w < m\} \in \nempfaclang[\nset{M}]$ of the sequences of costs.  

\begin{definition}[coalgebraic inference of cost-bounded reachability]
    \label{def:coinfCostBound}
    The coalgebraic inference of cost-bounded reachability is given by 
    \begin{itemize}
        \item The  system $c\colon X\rightarrow \mcfunc[\nset{M}]{X}$ with semantic domain $\subdistributionfunc([M]^{+})$ 
        \item  the requirement $d\colon\nsetbot{N}\rightarrow \dfautomata[\nset{M}]{\nsetbot{N}}$ defined in~\cref{def:dfas_for_costs}  with semantic domain $\nempfaclang[\nset{M}]$(defined in~\cref{ex:tau-semanticsDFA})
        \item The query $q\colon \subdistributionfunc([M]^{+})\times \nempfaclang[\nset{M}] \rightarrow \probinterval$ given by $q(\nu, L) \defeq \sum_{w\in L}\nu(w)$. 
    \end{itemize}
\end{definition}

Since~\cref{def:coinfCostBound} is the instance of the probabilistic inference defined in~\cref{ex:probinf_coalinf}, 
we obtain the coalgebraic product with the distributive law $\lambda_{X, Y}\colon \mcfunc[\nset{M}]{X}\times\dfautomata[\nset{M}]{Y}\rightarrow  \mcflagfunc{X\times Y}$ shown in~\cref{ex:dis_mc_dfa}, and the same semantic structure shown in~\cref{ex:ssProdMCDFA}. 
\begin{proposition}
    \label{prop:correctCostBound}
    The coalgebraic product $c\otimes_{\lambda} d$ of the cost-bounded reachability with $\tau_{S\otimes R}$ defined in~\cref{ex:ssProdMCDFA} meets the correctness criterion (\cref{thm:correctness_product}); 
    thus $c\otimes_{\lambda} d$ is correct w.r.t. the query $q$, and $\lfpoint[(x, N)]{c\otimes_{\lambda} d} = q(\lfpoint[x]{c}, \lfpoint[N]{d}) = \cbquery(\lfpoint[x]{c}, N)$. 
\end{proposition}
In fact, the known method of \emph{unfolding}~\cite{AndovaHK03,LaroussinieS05} for computing cost-bounded reachability coincides with the coalgebraic product construction.

\subsection{Costs Induced by Requirements}
\label{subsect:costReq}
We consider another cost-bounded reachability problem: costs are induced by requirements. 
Given an MC, the inference problem is given by 
\[ q_{cr}\colon   \distributionfunc(A^{+})\times {(\nat^{+})}^{A^{+}} \times \nat \rightarrow \probinterval, \quad q_{cr}(\nu, f, N) \defeq \sum_{w\in \supp(\nu),\, \big(\sum f(w)\big) < N} \nu(w). \]
Here, the requirement gives an assignment $f$ and a bound $N$, and asks that the accumulated sum $\sum f(w) \defeq f(w)_1+f(w)_2+\cdots $ of costs satisfies $\sum f(w) < N$ for each trace $w$. 
It is natural to assume that the assignment $f$ is finitely represented; here we define \emph{reward machines}:
    a reward machine is a coalgebra $d\colon Y\rightarrow (Y\times \nset{M})^A$.
Our approach to this inference can be summarised as follows: We create a new product DFA as the requirement from the reward machine $f$ and the bound $N$, and we take the coalgebraic product of the MC and the DFA.

The semantic domain for the reward machine $f$ is given by ${(\nsetbot{M}^{+})}^{A^{+}}$, where $\nsetbot{M}^+$ is the set $\nset{M}^+\uplus \{\bot\}$ with the ordering $\preceq$, defined, for any $z, z'\in \nsetbot{M}^+$, by $z\preceq z'$ if $z = z'$ or $z = \bot$.
We then define the semantics of the reward machines as the assignment of sequences of costs for each trace from each state $y$,
where the partial order $\preceq$ on ${(\nsetbot{M}^{+})}^{A^{+}}$ is given by $f\preceq f'$ iff $f(w) \preceq f'(w)$ for each trace $w\in A^+$; see~\cref{appendix:costReq} for the complete definitions of the semantics.  

We impose the condition of the bound $N$ by the coalgebraic product construction for reward machines and DFAs as follows: 

\begin{definition}
    Let $A$ and $B$ be finite sets, the \emph{distributive law} $\alpha_{Y, Z}\colon (Y\times B)^A\times \dfautomata[B]{Z}\rightarrow (Y\times Z\times \settorf)^A $ is given by
        $\alpha_{Y, Z}(\delta_1, \delta_2)(a) \defeq (y, z, t)$,
    where $(y, b) \defeq \delta_1(a)$, and $(z, t) \defeq \delta_2(b)$.  
\end{definition}
Given a reward machine $d_1$ and the requirement $d_2$ for cost-bounded reachability defined in~\cref{def:dfas_for_costs}, 
we obtain the following new requirement $d_3$ by the coalgebraic product construction with the distributive law $\alpha$, where $B\defeq \nset{M}$.   
\begin{definition}[requirement for reward machine and bound]
    Given a reward machine $d_1$, and a bound $N$, the \emph{requirement} $d_3\colon Y\times \nsetbot{N}\rightarrow \dfautomata[A]{Y\times \nsetbot{N}}$ is the DFA given by  $d_3\defeq d_1\otimes_{\alpha} d_2$,
    where the coalgebra $d_2\defeq \nsetbot{N}\rightarrow \dfautomata[\nset{M}]{\nsetbot{N}}$ is defined in~\cref{def:dfas_for_costs}. 
    Concretely, the requirement $d_3$ is given by  
    \begin{align*}
        d_3(y, i)(a) &\defeq \begin{cases}
            \big(y', i - j, \top \big) &\text{ if } i - j > 0,\\
            \big(y', \bot, \bot\big) &\text{ otherwise,} 
        \end{cases}\,
        &d_3(y, \bot)(a) \defeq (y', \bot, \bot),
    \end{align*} 
    where $(y', j) \defeq d_1(y)(a)$. 
\end{definition}

Then, the coalgebraic product with the MC $c\colon \mcfunc{X}$ and the requirement $d_3 \defeq d_1\otimes_{\alpha} d_2$ is correct w.r.t. the query $q_{cr}$ as follows. 
\begin{proposition}
    \label{prop:correctCostReq}
    The following equations hold: 
    \[
        \lfpoint[(x, y, N)]{c\otimes_{\lambda } d_3} = q_1(\lfpoint[x]{c}, \lfpoint[(y, N)]{d_1\otimes_{\alpha} d_2}) = q_1\big(\lfpoint[x]{c}, q_2(\lfpoint[y]{d_1}, \lfpoint[N]{d_2})\big)  = q_{cr}(\lfpoint[x]{c}, \lfpoint[y]{d_1}, N),
    \]
    where the query $q_1\colon \subdistributionfunc(A^{+})\times \nempfaclang[A] \rightarrow \probinterval$ is defined in~\cref{ex:probinf_coalinf}, and the query $q_2\colon{(\nsetbot{M}^{+})}^{A^{+}} \times \nempfaclang[\nset{M}]  \rightarrow \nempfaclang$ is defined by
       $q_2(f, L) \defeq \{ w\in A^{+}\mid \ f(w)\in L\}$. 
\end{proposition}
For the proof, it suffices to show that the coalgebraic product $d_3\defeq d_1\otimes_{\alpha} d_2$ satisfies the correctness criterion with the modality for DFA (defined in~\cref{ex:tau-semanticsDFA}) w.r.t. the query $q_2$, concluding that $\lfpoint[(y, N)]{d_3} = q_2(\lfpoint[y]{d_1}, \lfpoint[N]{d_2})$; see~\cref{appendix:costReq} for the detail.

\section{Case Study~\Rom{3}: Quantitative Inference on Probabilistic Systems and Safety Properties over Infinite Traces}
\label{sec:MCsNeverTerminate}
In this section, we consider probabilistic programs that do not terminate. 
Given a DFA as a requirement, we infer the likelihood of 
eventually accepting traces on probabilistic programs. 
Further, we show a translation from the quantitative inference of probabilistic programs that almost-surely terminate to that of probabilistic programs that do not terminate.

    An \emph{MC without target states} is a coalgebra $c\colon X\rightarrow \mcinffunc{X}$. 
    We additionally assume that MCs are \emph{finite branching}, that is, the support $\supp(P(x,\_))$ of the probabilisitic transition is finite for each $x\in X$. 
To formally define the semantics of MCs without target states, we introduce some measure-theoretic notions; we refer to~\cite{Baier08} as a gentle introduction to measure theory for MCs. 
    The \emph{cylinder set} $\cyset[]{w}$ of finite string $w \defeq a_1\cdots a_m$ over $A$ is the set of infinite strings defined by  
        $\cyset[]{w} \defeq \{w'\in A^{\omega} \mid w = a_1\cdots a_m \in \pref{w'} \}$, where $\pref{w'}$ is the set of prefixes of $w'$. 
    
\begin{definition}
    \label{def:sigma_algebra}
    The \emph{$\sigma$-algebra} $\salg{A}$ associated with $A$ is the smallest $\sigma$-algebra that contains all cylinder sets $\cyset[]{w}$.
    We write $\probmeasure(\salg{A})$ for the set of probability measures on $\salg{A}$. 
\end{definition}
For each $x\in X$, it has been known that there is the unique probability measure $\mathbb{P}_x$ on $\salg{A}$ such that 
$\mathbb{P}_x(\cyset[]{w}) = \sum_{l(x_1)\cdots l(x_n) = w, x_1 = x}\Pi_{i\in \nset{n-1}}P(x_{i},x_{i+1})$ for each $n\in \nat$ and $w\in A^{n}$; see e.g.~\cite{Baier08}.\footnote{The probability measure is defined over paths in~\cite{Baier08}, while we define it over traces. }

The appropriate definition of the semantic domain is not straightforward. 
We define the semantic domain $\Omega$ by the disjoint union of the set $ \probmeasure(\salg{A})$ of probability measures on $\salg{A}$ and  the set $ \biguplus_{n\in \nat} \distributionfunc(A^n)$ of distributions.
A non-trivial part of the semantics is the definition of the partial order $\preceq$ on the semantic domain $\Omega$:
the partial order $\preceq$ on $\Omega$ is given by the \emph{marginalization}. 

\begin{definition}[marginalization]
        Let $m,n\in \nat$ such that $m\leq n$, and $\sigma\in \distributionfunc(A^n)$.  
        We define the \emph{marginalization} $\mrg{\sigma}{m}\in  \distributionfunc(A^m)$ by $\mrg{\sigma}{m}(w) \defeq \sum_{w'\in A^n\land w\in \pref{w'}}\sigma(w')$. 
        Analogously, we define the \emph{marginalization} $\mrg{\sigma}{m}\in  \distributionfunc(A^m)$ for probability measures $\sigma\in \probmeasure(\salg{A})$ as well.  
\end{definition}
    
\begin{definition}[semantic domain]
    \label{def:truth-valueMC}
        The \emph{semantic domain} $\Omega$ is the set given by $\Omega \defeq \probmeasure(\salg{A})\uplus \biguplus_{n\in \nat} \distributionfunc(A^n)$. 
        The partial order $\preceq$ over  $\Omega$ is defined by the so-called Kolmogorov consistency condition. 
        Specifically, the order $\preceq$  is the smallest order that satisfies the following condition: 
        \begin{itemize}
            \item let $\sigma\in \distributionfunc(A^m)$, and $\sigma'\in \distributionfunc(A^n)$ and $m <  n$. 
                  If $\sigma = \mrg{\sigma'}{m}$, then $\sigma\prec \sigma'$. 
            \item let $\sigma\in \distributionfunc(A^m)$, and $\sigma'\in \probmeasure(\salg{A})$. 
            If $\sigma = \mrg{\sigma'}{m}$, then $\sigma\prec \sigma'$. 
        \end{itemize}
            
\end{definition}

The semantic domain forms an $\omega$-cpo with the least element by the Kolmogorov extension theorem. This is because the 
$\sigma$-algebra $(A, \mathcal{P}(A))$ is a Polish space (complete separable metric space); see e.g.~\cite{klenke2013probability} or~\cite{hairer2020markov} for comprehensive references.   

\begin{definition}[modality for MCs]
    \label{def:modalityMCs}
    The \emph{modality} $\tau_S: \mcinffunc{\Omega}\rightarrow \Omega$ of MCs over $A$ is given as follows.
    Let $\nu\in \distributionfunc(\Omega)$, $a\in A$, and $m\defeq \min\big(\{n\mid \sigma\in \supp(\nu)\text{ and }\sigma\in \distributionfunc(A^n)\}, \omega\big)$.
    If $m\in \nat$, then we define the distribution $\sigma_{m+1}\in \distributionfunc(A^{m+1})$ as follows: 
    \begin{align*}
        \sigma_{m+1}(w) \defeq \begin{cases}
            \sum_{\sigma\in \supp(\nu)} \nu(\sigma)\cdot \mrg{\sigma}{m}(w') &\text{ if $w = a\cdot w'$,}\\
            0 &\text{ if $a\not \in\pref{w}$.}
        \end{cases}
    \end{align*}
    The modality $\tau_S$ is then defined by $\tau_S(\nu, a) \defeq \sigma_{m+1}$. 
    If $m = \omega$, then we construct the sequence $(\sigma_n)_{n\in \nat}$ of distributions by the above construction. 
    We define the modality $\tau_S$ by $\tau_S(\nu, a)\defeq \bigvee (\sigma_n)_{n\in \nat}$; this is well-defined because the sequence $(\sigma_n)_{n\in \nat}$ is an $\omega$-chain (see~\cref{sec:proof_inf_trace} for the detail). 
\end{definition}

We assume that the requirement is given by a DFA $d\colon Y\rightarrow \dfautomata{Y}$.
The query $q$ is defined by the \emph{partition}, which is a standard technique in measure theory to obtain a monotone sequence.  

\begin{definition}[partition]
    \label{def:partition}
    Let $T\in \nempfaclang$. For each $n\in \nat\backslash\{0\}$, we inductively define the set $T^n\subseteq \mathcal{P}(A^n)$ as follows: 
    $T^1 \defeq \{w\in T \mid |w| = 1\}$, and $T^{n+1}\defeq \{w\in T\mid |w| = n+1\text{ and for any $w'\in \pref{w}$, $w'\not\in \cup_{i\in [1, n]} T^i$}\}$.
\end{definition}

\begin{definition}[query]
    \label{def:query_mc_not-terminate}
The query $q\colon \Omega \times \nempfaclang\rightarrow \probinterval$ is given by 
\begin{align*}
    q\big(\sigma,T\big) \defeq  \begin{cases}
        \sum_{i\in [1, n]}\mrg{\sigma}{i}(T^i) &\text{ if $\sigma\in \distributionfunc(A^n)$,}\\
        \sigma(T^{\omega}) &\text{ if $\sigma\in \probmeasure(\salg{A})$,}
    \end{cases}
\end{align*}
where $T^n$ is defined in~\cref{def:partition}, and  $T^{\omega} \defeq \bigcup_{l\in T} \cyset[]{l}$.
\end{definition}

We obtain the coalgebraic product that is similar to the one in~\cref{ex:dis_mc_dfa}. 

\begin{definition}
    \label{def:dist_MC_not-terminate}
    The distributive law $\lambda_{X, Y}\colon \mcinffunc{X}\times \dfautomata{Y}\rightarrow \dmcflagfunc{X\times Y}$ is given by 
    \begin{align*}
        \lambda_{X, Y}(\nu, a, \delta)(x, y) &\defeq \begin{cases}
            \nu(x) &\text{if $\delta(a) = (y, \bot)$,}\\
            0 &\text{otherwise,}
        \end{cases}
        & \lambda_{X, Y}(\nu, a, \delta)(\star) \defeq \begin{cases}
            1 &\text{if $\delta(a) = (\_, \top)$,}\\
            0 &\text{otherwise.}
        \end{cases}
    \end{align*}
\end{definition}

\begin{definition}
    The modality $\tau_{S\otimes R}\colon  \dmcflagfunc{\probinterval}\rightarrow \probinterval$ is given by 
    \begin{align*}
        \tau_{S\otimes R}(\nu) \defeq \nu(\star) + \sum_{p\in \probinterval} \nu(p)\cdot p. 
    \end{align*}
\end{definition}

The coalgebraic product satisfies the weaker correctness criterion in~\cref{prop:weaker_correctness_criterion}; see~\cref{sec:proof_inf_trace} for the proof.

\begin{proposition}
    The product  with  the modality $\tau_{S\otimes R}$ is correct with respect to the query $q$. 
\end{proposition}
We note that there is a simple translation from the quantitative inference on MCs that almost surely terminates given in~\cref{sec:coalgebraic_inference} to the quantitative inference on MCs that do not terminate; see~\cref{sec:translation} for the detail.

\section{Case Study \Rom{4}: Optimization with Weighted Programs}
\label{sec:weighted_programming}
In this section we show that optimization problems for weighted programs based on the tropical semiring, introduced in~\cref{sec:ski_snowboard}, are an instance of coalgebraic inference; see~\cite{CohenSS11,BrunelGMZ14,GaboardiKOS21,BatzGKKW22,BelleR20} for general references on weighted programming. 
We also show that the associated product construction is correct.

The operational semantics of weighted programs is presented in terms of suitably defined coalgebras; we call them \emph{weighted transition systems} in this paper. 
Formally, a \emph{weighted transition system} is a coalgebra $c\colon X\rightarrow \ltsfuncsemiring{X}{\nat}$.
The \emph{semantics domain of weighted transition systems} is the set $\powersetfunc(A^{+}\times \nat)$ of records consisting of (i) a trace $w\in A^{+}$ to the target state and (ii) a possible accumulated weight $r\in \nat$ along the trace $w$. 
Note that each trace $w$ can correspond to multiple paths whose accumulated weights are different. 

\begin{definition}[semantic structure for weighted transition systems]
The \emph{semantic structure} is $\big((\powersetfunc(A^{+}\times \nat), \subseteq),\, \modsys{\tau}\big)$, where $\modsys{\tau}\colon \ltsfuncsemiring{ \powersetfunc(A^{+}\times \nat)}{\nat}\rightarrow  \powersetfunc(A^{+}\times \nat)$ is given by  
\begin{align*}
    \modsys{\tau}(T_1)\defeq \{ (a,\, m) \mid (\star, a, m) \in T_1\} \cup \{ (a\cdot w,\,m+n)\mid (T_2, a, m) \in T_1,\,  (w, n)\in T_2\}.
\end{align*}

\end{definition}

For inference with weighted programs, we use \emph{non-deterministic finite automata (NFAs)} as the requirement; the weighted transition system has non-deterministic branching---unlike probabilistic branching---that is compatible not only with DFAs but also with NFAs. 
Formally, a \emph{non-deterministic finite automaton (NFA)} is a coalgebra $d\colon Y\rightarrow \exepnfautomata{Y}$, where the underlying set $Y$ is finite. 
The semantics of NFAs is the recognized languages of them excluding the empty string $\epsilon$.
\begin{definition}[semantic structure for NFAs]
The \emph{semantic structure} is $\big((\nempfaclang, \subseteq), \modspec{\tau}\big)$, where the modality $\modspec{\tau}\colon \exepnfautomata{\nempfaclang}\rightarrow \nempfaclang$ is given by  
\begin{equation*}
    \modspec{\tau}(\delta) \defeq \{ a \mid  (\_,\top) \in \delta(a)\} \cup \{a\cdot w \mid (T,\_)\in \delta(a),\, w\in T \}.
\end{equation*}
\end{definition}

The inference with weighted programs is computing the minimum cost to reach the target state along with a run accepted by the requirement; this is a coalgebraic inference (implicitly) induced by the tropical semiring $(\natinf = \nat + \{\infty\}, \min, +)$~\cite{pin1998tropical}. 

\begin{definition}[query for weighted transition systems and NFAs]
\label{def:query_wts_nfas}
    Let $T\in \powersetfunc(A^{+}\times \nat)$ and $L\in \nempfaclang$. 
    The query  $q\colon\powersetfunc(A^{+}\times \nat)\times \nempfaclang\rightarrow\natinf$ is given by 
   $ q(T, L)\defeq  \min \{m \mid w\in L,\, (w, m)\in T \}$.
Here, we define $\min \emptyset\defeq \infty$ as usual. 
    
\end{definition}

The coalgebraic product construction also works well for weighted inference.
\begin{definition}
\label{def:dist_weighted}
    The \emph{distributive law} $\lambda$ from   $\ltsfuncsemiring{\_}{\nat}$ and $\exepnfautomata{\_}$ to $\suptlflagfunc{\_}{\nat}$ is:\footnote{In~\cref{def:dist_weighted,def:dist_weighted_WMA}, we implicitly assume $b\in \settorf$, $x\in X$, and $y\in Y$ for readability.}
\begin{equation*}
    \lambda_{X, Y}(T, \delta) \defeq \big\{(b,m) \ \big| \ (\star, a, m)\in T \text{ and }(\_, b)\in\delta(a)\big\}\cup \big\{ (x, y,  m) \ \big| \ (x, a, m)\in T,\, (y,\_)\in \delta(a)\big\}.
\end{equation*}
\end{definition}

The semantics of the coalgebraic product $c\otimes_{\lambda} d$ induced by $\lambda$ directly computes the minimum cost to reach the target state without any requirements. 

\begin{definition}
    \label{def:sem_weight}
Let $c$ be a weighted transition system, and $d$ be a  NFA. 
    The \emph{semantic structure} is  $\big((\natinf,  \leq_o), \modprod{\tau}\big)$, where the order $\leq_o$ is given by $m \leq_o n$ if $m\geq n$, and the modality $\modprod{\tau}\colon\suptlflagfunc{\natinf}{\nat}\rightarrow \natinf$ is given by 
    $\modprod{\tau}(T)\defeq\min\big( \{ m \mid (\top, m) \in T\} \cup \{ m+ n\mid (n, m)\in T\} \big)$.
\end{definition}

The coalgebraic product construction satisfies the correctness criterion shown in~\cref{thm:correctness_product}: the inference $ q(\lfpoint[x]{c}, \lfpoint[y]{d})$ in~\cref{def:query_wts_nfas} coincides with the semantics $\lfpoint[(x, y)]{c\otimes_{\lambda} d}$,
\iffull
see~\cref{sec:proof_weight_nfas}. 
\else 
see~\cite[Appendix A]{Watanabe24}. 
\fi

\begin{proposition}
 The product with the modality $\tau_{S\otimes R}$ is correct with respect to the query $q$ given in~\cref{def:query_wts_nfas}. 
\end{proposition}

\section{Case Study~\Rom{5}: Quantitative Inference  with Weighted Regular Safety Properties}
\label{sec:quantitativeInferenceWeightedSystemRequirement}
To demonstrate the generality of our approach, we present a new product construction of weighted transition systems introduced in~\cref{sec:weighted_programming} and weighted Mealy machines (WMMs), where the latter impose additional penalties on the traces. 
The inference asks for the minimum cumulative weight of accepting traces. 
Specifically, for each accepting trace $w$, there may be multiple accumulated weights along the trace on a given system and requirement, respectively. 
The minimum cumulative weight of accepting traces is defined by the minimum of the sum $m_{w} + n_{w}$, where $w$ is an accepting trace,  $m_{w}$ and $n_{w}$ are possible accumulated weights along the accepting trace $w$ on the system and the requirement, respectively. 

A \emph{weighted Mealy machine (WMM)} is a coalgebra $d\colon Y\rightarrow \wmealyfunc{Y}$, where the underlying set $Y$ is finite. 
The \emph{semantic domain of WMMs} is the set $\powersetfunc(A^{+}\times \nat)$ of records consisting of (i) an accepting (non-empty) trace $w\in A^{+}$, and  (ii) a possible accumulated weight $r\in \nat$ along the trace $w$.

\begin{definition}[semantic structure for WMMs] The \emph{semantic structure} is $\big((\powersetfunc(A^{+}\times \nat), \subseteq),\, \modspec{\tau}\big)$, where $\modsys{\tau}\colon \wmealyfunc{\powersetfunc(A^{+}\times \nat)}\rightarrow \powersetfunc(A^{+}\times \nat) $ is given by 
\begin{align*}
    \modspec{\tau}(\delta) &\defeq \big\{ (a, m) \, \big|\, (\_, \top, m)\in \delta(a)  \big\}\cup \big\{ (a\cdot w, m+n) \, \big|\, (T, \_, m)\in \delta(a),\, (w, n)\in T  \big\}.
\end{align*}
\end{definition}

Then, the quantitative inference, which is an extension of the one in~\cref{sec:weighted_programming}, deals not only with the cost imposed by a given weighted transition system, but also with the penalty imposed by a WMM. 
\begin{definition}[query for weighted transition systems and WMMs]
    \label{def:query_wts_WMAs}
Let $T,L\in \powersetfunc(A^{+}\times \nat)$. The query $q\colon \powersetfunc(A^{+}\times \nat)\times \powersetfunc(A^{+}\times \nat)\rightarrow \natinf$ is given by 
    $q(T, L) \defeq \min\{ m + n \mid (w,m)\in T, (w, n)\in L\}$.
\end{definition}

The coalgebraic product construction integrates the costs on systems and the penalties on requirements: the following is the corresponding distributive law. 

\begin{definition}
\label{def:dist_weighted_WMA}
    The distributive law $\lambda$ from   $\ltsfuncsemiring{\_}{\nat}$ and  $\wmealyfunc{\_}$ to $\suptlflagfunc{\_}{\nat}$ is given by:
\begin{align*}
    \lambda_{X, Y}(T, \delta) &\defeq \big\{(b,m+n) \ \big| \ (\star, a, m)\in T \text{ and }(\_, b,n)\in\delta(a)\big\}\\
&\cup \big\{ (x, y,  m+n) \ \big| \ (x, a, m)\in T,\, (y,\_,n)\in \delta(a)\big\}.
\end{align*}
\end{definition}

With the semantics of the coalgebraic product $c\otimes_{\lambda} d$ induced by $\lambda$ presented in~\cref{def:sem_weight}, 
the coalgebraic product $c\otimes_{\lambda} d$ induced by $\lambda$ satisfies the correctness criterion (\cref{thm:correctness_product}), and thus it is correct: 
see~\cref{sec:proof_weight_WMA} for the proof. 
\begin{proposition}
    The product with the modality $\tau_{S\otimes R}$ defined in~\cref{def:sem_weight}  is correct with respect to the query $q$ given in~\cref{def:query_wts_WMAs}. 
\end{proposition}

\section{Related Work}
\label{sec:relatedwork}

Below we discuss related work on (i) unifying framework for verification, and (ii) temporal inference in weighted programs. There is a wide variety of results and literature on probabilistic verification, which we do not cover here beyond what was mentioned in the introduction: our main distinguishing feature is the generality of our approach.

\paragraph{Unifying frameworks.}
The theory of \emph{coalgebras} brings a generic representation of transition systems. The techniques that we use fit into a tradition of coalgebraic definitions of trace semantics and determinization (e.g.,~\cite{HasuoJS07,DBLP:journals/corr/abs-1302-1046,Jacobs0S15,DBLP:journals/tocl/BonsangueMS13,KlinR16,RotJL21})---the idea to use distributive laws with three different functors is new, and is precisely the starting point of product constructions in the context of verification. We refer to~\cite{Jacobs16} for an overview of the scope and applications of coalgebra, and restrict our attention to general coalgebraic frameworks related to verification.

C\^irstea proposes a coalgebraic approach to linear-time quantitative behaviour, including $\omega$-regular property~\cite{Cirstea14,Cirstea17a}. The current paper instead focuses on the algorithmic perspective provided by product constructions, as a practical technique for solving inference queries; supporting $\omega$-regular properties is future work.

There are several lines of work based on coalgebras for verification on \emph{minimization}~\cite{lmcs/WissmannDMS19,JacobsW23}, \emph{liveness checking}~\cite{UrabeHH17}, \emph{$\omega$-regular automata}~\cite{UrabeH18,CianciaV19},~\emph{progress measures}~\cite{HasuoSC16},~\emph{CTL}~\cite{KojimaCMH24} and \emph{PDR}~\cite{DBLP:conf/cav/KoriUKSH22,KoriABBGH23}, using lattice theory.
Product constructions and their correctness at a coalgebraic level are an original contribution, to the best of our knowledge. 

The work~\cite{AguirreKK22} gives a generic compositional framework for weakest preconditions via predicate transformers~\cite{Dijkstra75}. It is based on fibrations ~\cite{Jacobs2001}; we use fibrations only implicitly (see~\cref{sec:lifting}). The recent work~\cite{KoriWRK24} shows a sufficient condition for compositionality of a general form of bisimilarity (so-called \emph{codensity bisimilarity}) with respect to abstract parallel composition operators, defined as distributive laws of products over a behaviour functor $F$. Both the technical development and the examples differ significantly from the current paper. In particular, we do not make use of the codensity lifting but use a least fixed point semantics.

\paragraph{Temporal inference in weighted programs.}

Concerning weighted systems, the work~\cite{BaierKKW14} gives a theory of weighted inference with temporal logic and monitored weight assertions. They cover a variety of systems including MDPs, and they present an extended LTL with new operators that introduce constraints on accumulated weights along traces. In this paper, we deal with an inference based on the coalgebraic product construction with weighted systems and weighted requirements, where the weights induced by systems and requirements interact with each other.

The work~\cite{DrosteR16} presents a general theory of \emph{weighted linear dynamic logic (weighted LDL)} and prove that the equivalence problem for weighted LDL formulas is decidable by constructing equivalent \emph{weighted B\"{u}chi automata}~\cite{DrosteM12}. Conducting inference with \emph{weighted $\omega$-regular specifications}~\cite{DrosteM12} based on product constructions indeed can be an interesting future work. 

The work~\cite{DodaroFG22} studies weighted traces for \emph{$\mathrm{LTL}_f$}~\cite{GiacomoV13}, which is a temporal logic on finite traces. They introduce a generic weighted inference framework based on \emph{valuation structures} that induce weighted structures.

\section{Conclusion and Future work}
\label{sec:conc_and_future}

We introduced a general coalgebraic framework for temporal inference. The key notion that of a \emph{product construction}, defined via distributive laws, and its correctness. A correct product construction essentially allows us to reduce inference over the output traces of a program to inference on the output distribution, for which we can use efficient inference methods. 
Our framework is motivated by probabilistic programming, but due to the coalgebraic underpinnings it applies to a broad variety of systems---in particular, the framework led us to define a correct product construction for an original inference problem: that of weighted regular safety requirements in weighted programming.

\paragraph{Future work} 
Our current method works at the level of the \emph{operational semantics} of probabilistic programs. The product with the requirement yields a (probabilistic) transition system. As discussed in the introduction, this provides a stepping stone to using off-the-shelf inference methods for probabilistic programs. What is still missing is a \emph{syntactic} presentation of product programs. We envisage that such a presentation may follow by combining our framework with that of mathematical operational semantics~\cite{TuriP97,Klin11}, which uses distributive laws to present rule formats for operational semantics and is based on (co)algebras. An intermediate step is to work at the level of probabilistic control flow graphs (e.g.,~\cite{ChatterjeeFG16,TakisakaOUH18,AgrawalC018,HasuoOESCK24}). 

It would also be interesting to see whether we can get a canonical semantics for the product from two final coalgebras for a system and a requirement. 
This is challenging for the existing approaches, including the bialgebraic approach~\cite{TuriP97,Jacobs06,Klin07}, because we have to build our framework in Kleisli categories for the final coalgebra semantics~\cite{HasuoJS07}, and the inference problem requires three different behaviour functors. 

Our current work does not support continuous distributions, such as those supported in Church~\cite{GoodmanMRBT08} and Anglican~\cite{TolpinMW15}, and seeking an extension of our framework for continuous distributions is another future direction.

\bibliographystyle{ACM-Reference-Format}
\bibliography{mybib}

\fi
\iffull

\appendix

\section{Omitted Definitions}
\label{sec:omitted_definitions}
\begin{definition}[$\omega$-cpo]
An \emph{$\omega$-cpo} is a partially ordered set $(\Omega, \leq)$ such that (i) there is a least element $\bot$, and (ii) for each sequence $(t_i)_{i\in \nat}$ such that $t_i \leq t_{i+1}$ for $i\in \nat$ (also called \emph{$\omega$-chain}), a supremum $\bigvee_{i\in \nat} t_i$ exists in $\Omega$.
\end{definition}

\begin{definition}[$\omega$-continuous]
Let $(\Omega, \leq)$ be an $\omega$-cpo. A function $f: \Omega\rightarrow \Omega$ is \emph{$\omega$-continuous} if for each $\omega$-chain $(t_i)_{i\in \nat}$, the equality $\bigvee_{i\in \nat}(f(t_i)) = f\big( \bigvee_{i\in \nat} t_i\big)$ holds. 
\end{definition}

Note that every $\omega$-continuous function $f$ is monotone, that is, for any $t, t'$, if $t\leq t'$, then $f(t) \leq f(t')$.

\section{Omitted Propositions and Proofs}
\label{sec:omitted_proofs}
\subsection{Proof of~\cref{thm:correctness_product}}
\label{proof:thm_correctness_product}
\begin{proof}
        Let $c\colon X\rightarrow F_S(X)$ and $d\colon Y\rightarrow F_R(Y)$ be coalgebras. 
        Since $\predtran{\tau_{S\otimes R}, c\otimes_{\lambda} d},\predtran{\modsys{\tau}, c},\predtran{\modspec{\tau}, d}$ are $\omega$-continuous, it is enough to prove that
        \begin{equation*}
            \bigvee_{i\in \nat} (\predtran{\modprod{\tau}, c\otimes_{\lambda} d})^{i}(\bot) = q\circ \big(\bigvee_{j\in \nat} (\predtran{\modsys{\tau}, c})^{j}(\bot)\big)\times \big(\bigvee_{k\in \nat} (\predtran{\modspec{\tau}, d})^{k}(\bot)\big).
        \end{equation*} The query $q$ is $\omega$-continuous, therefore
        \begin{equation*}
            q\circ \big(\bigvee_{j\in \nat} (\predtran{\modsys{\tau}, c})^{j}(\bot)\big)\times \big(\bigvee_{k\in \nat} (\predtran{\modspec{\tau}, d})^{k}(\bot)\big) = \bigvee_{n\in \nat} q\circ \big((\predtran{\modsys{\tau}, c})^{n}(\bot)\big)\times \big((\predtran{\modspec{\tau}, d})^{n}(\bot)\big).
        \end{equation*}
        Thus, it is sufficient to prove the following equality for $n\in \nat$ by induction:
        \begin{equation*}
            (\predtran{\modprod{\tau}, c\otimes_{\lambda} d})^{n}(\bot) = q\circ \big((\predtran{\modsys{\tau}, c})^{n}(\bot)\big)\times \big((\predtran{\modspec{\tau}, d})^{n}(\bot)\big).
        \end{equation*}
        For the base case $(n=0)$, the equality $\bot = q\circ \bot\times  \bot$ holds, since $q$ preserves the least elements. 
        For the step case $(n=k+1)$, suppose that the equality \[(\predtran{\modprod{\tau}, c\otimes_{\lambda} d})^{k}(\bot) = q\circ \big((\predtran{\modsys{\tau}, c})^{k}(\bot)\big)\times \big((\predtran{\modspec{\tau}, d})^{k}(\bot)\big)\] holds. Let \[u_{S\otimes R} \defeq (\predtran{\modprod{\tau}, c\otimes_{\lambda} d})^{k}(\bot),\quad u_S \defeq (\predtran{\modsys{\tau}, c})^{k}(\bot),\quad\text{ and }u_R \defeq (\predtran{\modspec{\tau}, d})^{k}(\bot).\] Then, the following diagram commutes by the naturality of $\lambda$, the induction hypothesis, and Eq.~\cref{eq:coherence_evaluation}:
        \begin{equation*}
             \begin{tikzcd}
                 \fsys{F}(X)\times \fspec{F}(Y)  \arrow[rrr,"\fsys{F}(u_S)\times \fspec{F}(u_R)"] \arrow[d, "\lambda_{X, Y}"]
                 &&&\fsys{F}(\tvdsys{\Omega})\times \fspec{F}(\tvdspec{\Omega})  \arrow[rr,"\modsys{\tau}\times \modspec{\tau}"]  \arrow[d, "\lambda_{\tvdsys{\Omega}, \tvdspec{\Omega}}"]
                 &&\tvdsys{\Omega}\times \tvdspec{\Omega}\arrow[dd, "q"]\\
                 F_{S\otimes R}(X\times Y) \arrow[rrr,"F_{S\otimes R}(u_S\times u_R)"]
                 \arrow[rrrd, "F_{S\otimes R}(u_{S\otimes R})"]
                 &&& F_{S\otimes R}(\tvdsys{\Omega}\times\tvdspec{\Omega}) \arrow[d,"F_{S\otimes R}(q)"]
                 &&\\
                 &&& F_{S\otimes R}(\Omega_{S\otimes R}) 
                 \arrow[rr, "\tau_{S\otimes R}"]
                 && \Omega_{S\otimes R}.
            \end{tikzcd}
        \end{equation*}
        By definition of $\predtran{\modprod{\tau}, c\otimes_{\lambda} d}$, $\predtran{\modsys{\tau}, c}$, and $\predtran{\modspec{\tau}, d}$, we conclude that 
        \begin{equation*}
            (\predtran{\modprod{\tau}, c\otimes_{\lambda} d})^{k+1}(\bot) = q\circ \big((\predtran{\modsys{\tau}, c})^{k+1}(\bot)\big)\times \big((\predtran{\modspec{\tau}, d})^{k+1}(\bot)\big). \qedhere
        \end{equation*}
        
\end{proof}

\subsection{MCs and DFAs}
\begin{lemma}\label{prop:omega_cpo_e_mc_dfas}
The query $q\colon\fsublang\times \nempfaclang\rightarrow \probinterval$ for MCs and DFAs is $\omega$-continuous. 
\end{lemma}
\begin{proof}
Let  $(\nu_i)_{i\in \nat}$ and $(T_j)_{j\in \nat}$ be $\omega$-chains in $\fsublang$ and $\nempfaclang$. Then, we need to prove the following equality: 
\begin{equation*}
   q\Big(\big(\bigvee_{i\in \nat} \nu_i\big), \big(\bigvee_{j\in \nat} T_j\big)\Big) \defeq  (\bigvee_{i\in \nat} \nu_i)\big(\bigvee_{j\in \nat} T_j\big) = \bigvee_{k\in \nat}q(\nu_k, T_k) \defeq \bigvee_{k\in \nat} \nu_{k}(T_k). 
\end{equation*}
By definition, 
\begin{align*}
    &(\bigvee_{i\in \nat} \nu_i)\big(\bigvee_{j\in \nat} T_j\big) = \lim_{i\rightarrow \infty } \nu_i\big( \bigcup_{j\in \nat} T_j \big), &\bigvee_{k\in \nat} \nu_{k}(T_k) = \lim_{k\rightarrow \infty} \nu_k(T_k).
\end{align*}
By mototone convergence theorem, 
\begin{equation*}
   \lim_{i\rightarrow \infty } \nu_i\big( \bigcup_{j\in \nat} T_j \big) = \lim_{i\rightarrow \infty } \lim_{j\rightarrow \infty }\nu_i(T_j). 
\end{equation*}
Since $\big(\nu_{i}(T_j)\big)_{(i, j)\in \nat\times \nat}$ is a monotone double sequence,  we conclude that 
\begin{equation*}
   (\bigvee_{i\in \nat} \nu_i)\big(\bigvee_{j\in \nat} T_j\big) = \lim_{i\rightarrow \infty } \lim_{j\rightarrow \infty }\nu_i(T_j) = \lim_{k\rightarrow \infty}\nu_{k}(T_k)= \bigvee_{k\in \nat}\nu_{k}(T_k). 
\end{equation*}
\end{proof}

\begin{lemma}
\label{lem:commute_mc_dfas}
Let $S = \big((\fsublang, \preceq), \modsys{\tau}\big)$, $R = \big((\faclang, \subseteq), \modspec{\tau}\big)$ be semantic structures given in~\cref{ex:tau-semanticsMC,ex:tau-semanticsDFA}, and $q$ be the query given in~\cref{ex:probinf_coalinf}. 
The following diagram commutes: 
\begin{equation*}
\adjustbox{scale=0.9,center}{%
\begin{tikzcd}
\mcfunc{\fsublang}\times \dfautomata{\nempfaclang}\arrow[rr,"\modsys{\tau}\times \modspec{\tau}"] \arrow[d,"\lambda_{\fsublang,\, \nempfaclang}"]
            & 
            &\fsublang\times\nempfaclang\arrow[d,"q"]\\
\mcflagfunc{\fsublang\times\nempfaclang }\arrow[r,"\mcflagfunc{q}"]
            & \mcflagfunc{\probinterval} \arrow[r, "\modprod{\tau}"]
            & \probinterval
\end{tikzcd}
}
\end{equation*}

\end{lemma}
\begin{proof} Let $(\nu, a, \delta) \in \mcfunc{\fsublang}\times \dfautomata{\nempfaclang}$, and $(T, b) \defeq \delta(a)$. We prove the statement by the following direct calculations:
    \begin{align*}
        &\big(q \circ \modsys{\tau}\times \modspec{\tau}\big)(\nu, a, \delta)\\
        &= \sum_{a'\in \modspec{\tau}(\delta) }\modsys{\tau}(\nu, a)(a') + \sum_{a'\cdot w \in  \modspec{\tau}(\delta)}\modsys{\tau}(\nu, a)(a'\cdot w)   \\
        & = \nu(\star) \cdot \dirac{b = \top} + \sum_{\nu'\in \supp(\nu)} \nu(\nu')\cdot \nu'(T), \text{ and }\\
        &\big(\modprod{\tau}\circ \mcflagfunc{q}\circ \lambda_{\fsublang,\, \nempfaclang}\big)(\nu, a, \delta)\\
        &= \big(\mcflagfunc{q}\circ \lambda_{\fsublang,\, \nempfaclang}\big)(\nu, a, \delta)(\top) + \sum_{r\in \probinterval} \Big(\big(\mcflagfunc{q}\circ \lambda_{\fsublang,\, \nempfaclang}\big)(\nu, a, \delta)(r)\Big)\cdot r\\
        &= \nu(\star) \cdot \dirac{b= \top}  +  \sum_{r\in \probinterval}\big( \mcflagfunc{q}\circ \lambda_{\fsublang,\, \nempfaclang}(\nu, a, \delta)(r)\big)\cdot r\\
        &= \nu(\star) \cdot \dirac{b = \top}  +  \sum_{(\nu', T')\in \fsublang\times \nempfaclang} \big(\lambda_{\fsublang,\, \nempfaclang}(\nu, a, \delta)(\nu', T')\big)\cdot \nu'(T')\\
        &= \nu(\star) \cdot \dirac{b = \top}  +  \sum_{\nu'\in \supp(\nu)} \nu(\nu') \cdot \nu'(T). 
    \end{align*}
\end{proof}

\subsection{Markov Reward Models}
\label{subsec:MRM}
\begin{definition}
The semantic structure for MCs with rewards is given as follows: 
    \begin{itemize}
        \item a semantic domain $\mathbf{\Omega} \defeq (\subdistributionfunc(A^{+}\times \nat), \preceq)$, ordered pointwise,
        \item a modality $\tau_S: \distributionfunc\big(\subdistributionfunc(A^{+}\times \nat) \times \nat \times A+\{\star\}\big) \times \nat \times A \rightarrow \subdistributionfunc(A^{+}\times \nat)$ is given by 
        \begin{align*}
            \tau_S(\nu, n, a)(w, m) &\defeq \begin{cases}
                \nu(\star) &\text{ if $w = a$ and $n = m$,}\\
                \sum_{\nu'\in \supp(\nu)}\nu(\nu')\cdot \nu'(w', m-n)&\text{ if $w = a\cdot w'$ and $n \leq m$,}\\
                0 &\text{ otherwise, }
            \end{cases}
        \end{align*}
    \end{itemize}
\end{definition}

\begin{definition}
    The semantic structure for the product for MCs with rewards is given as follows: 
    \begin{itemize}
        \item a semantic domain  $\mathbf{\Omega} \defeq (\probinterval\times \nonegreal, \leq\times \leq)$, 
        \item a modality $\tau_{S\otimes R}: \distributionfunc(\probinterval\times \nonegreal+\settorf) \times \nat\rightarrow \probinterval\times \nonegreal$ is given by $\tau_{S\otimes R}(\nu, n) \defeq (p, r)$, where
        \begin{align*}
            p &\defeq \nu(\top) + \sum_{(p', r')\in \supp(\nu)} \nu(p', r') \cdot p',\\
            r &\defeq n \cdot \nu(\top) + \sum_{\substack{(p', r')\in \supp(\nu)}} \nu(p', r') \cdot \big(p'\cdot n + r'\big).
        \end{align*}
    \end{itemize}
\end{definition}

\begin{proposition}
The product for MCs with rewards is correct.
\end{proposition}
\begin{proof}
    The query $q$ preserves the least elements trivially. Let $(\nu_i)_{i\in \nat}$ and $(L_j)_{j\in \nat}$ be $\omega$-chains on $\subdistributionfunc(A^{+}\times \nat)$ and $\nempfaclang$. We prove the following equality: 
    \begin{align*}
        \sum_{n\in \nat} (\lim_{i \rightarrow \infty} \nu_i)\Big(\big\{ (w, n) \mid w\in \cup_{j\in \nat} L_j\big\} \Big) = \lim_{k\rightarrow \infty} \sum_{n\in \nat} \nu_k\Big( \big\{ (w, n) \mid w\in L_k \big\}\Big).
    \end{align*}
    Since $\lim_{i\rightarrow \infty}\nu_i$ is sub-distribution, the following equation holds by the monotone convergence theorem: 
    \begin{align*}
        \sum_{n\in \nat} (\lim_{i \rightarrow \infty} \nu_i)\Big(\big\{ (w, n) \mid w\in \cup_{j\in \nat} L_j\big\} \Big)= \sum_{n\in \nat} \lim_{j\rightarrow \infty} (\lim_{i \rightarrow \infty} \nu_i)\Big(\big\{ (w, n) \mid w\in L_j\big\} \Big).
    \end{align*}
    By monotonicity, 
    \begin{align*}
       \sum_{n\in \nat} \lim_{j\rightarrow \infty} (\lim_{i \rightarrow \infty} \nu_i)\Big(\big\{ (w, n) \mid w\in L_j\big\} \Big) &= \lim_{j\rightarrow \infty} \sum_{n\in \nat} (\lim_{i\rightarrow \infty} \nu_i)\Big(\big\{ (w, n) \mid w\in L_j\big\} \Big)\\
       &= \lim_{j\rightarrow \infty} \lim_{i\rightarrow \infty} \nu_i\Big(\big\{ (w, n) \mid w\in L_j,\, n\in \nat\big\} \Big)\\
       &= \lim_{k\rightarrow \infty} \nu_k\Big(\big\{ (w, n) \mid w\in L_k,\, n\in \nat\big\} \Big)\\
       &= \sum_{n\in \nat}\lim_{k\rightarrow \infty} \nu_k\Big(\big\{ (w, n) \mid w\in L_k\big\} \Big). 
    \end{align*}
    Next, we prove the following equality: 
    \begin{align*}
        \sum_{n\in \nat} n\cdot (\lim_{i \rightarrow \infty} \nu_i)\Big(\big\{ (w, n) \mid w\in \cup_{j\in \nat} L_j\big\} \Big) = \lim_{k\rightarrow \infty} \sum_{n\in \nat} n\cdot \nu_k\Big( \big\{ (w, n) \mid w\in L_k \big\}\Big).
    \end{align*}
    By monotone convergence theorem and monotonicity, 
    \begin{align*}
        \sum_{n\in \nat} n\cdot (\lim_{i \rightarrow \infty} \nu_i)\Big(\big\{ (w, n) \mid w\in \cup_{j\in \nat} L_j\big\} \Big) &= \sum_{n\in \nat} \lim_{j\rightarrow \infty} n\cdot (\lim_{i \rightarrow \infty} \nu_i)\Big(\big\{ (w, n) \mid w\in L_j\big\} \Big)\\
        &=  \lim_{j\rightarrow \infty} \sum_{n\in \nat} n\cdot (\lim_{i \rightarrow \infty} \nu_i)\Big(\big\{ (w, n) \mid w\in L_j\big\} \Big)\\
        &= \lim_{j\rightarrow \infty} \sum_{n\in \nat} n\cdot \sum_{w\in L_j} \lim_{i \rightarrow \infty} \nu_i(w, n)\\
        &= \lim_{j\rightarrow \infty} \lim_{i \rightarrow \infty} \sum_{n\in \nat} n\cdot \sum_{w\in L_j} \nu_i(w, n)\\
        &= \lim_{k\rightarrow \infty} \sum_{n\in \nat} n\cdot \sum_{w\in L_k} \nu_k(w, n).
    \end{align*}

    Next, we prove that the following diagram commutes: 
    \begin{equation*}
    \adjustbox{scale=0.9,center}{%
    \begin{tikzcd}
    \distributionfunc(\subdistributionfunc(A^{+}\times \nat) + \{\star\}) \times \nat \times A\times \dfautomata{\nempfaclang}\arrow[rr,"\modsys{\tau}\times \modspec{\tau}"] \arrow[d,"\lambda_{\subdistributionfunc(A^{+}\times \nat), \nempfaclang}"]
            & 
            &\fsublang\times\nempfaclang\arrow[d,"q"]\\
    \distributionfunc(\subdistributionfunc(A^{+}\times \nat)\times\nempfaclang + \settorf)\times \nat\arrow[r,"\distributionfunc(q + \settorf)\times \nat"]
            & \distributionfunc(\probinterval\times \nonegreal + \settorf)\times \nat \arrow[r, "\modprod{\tau}"]
            & \probinterval\times \nonegreal
    \end{tikzcd}
    }
    \end{equation*}
    For $(\nu, n, a, \delta)\in  \distributionfunc(\subdistributionfunc(A^{+}\times \nat) + \{\star\}) \times \nat \times A\times \dfautomata{\nempfaclang}$, $q\circ (\tau_S\times \tau_R)(\nu, n, a, \delta) =  (p_1, r_1)$, where 
    \begin{align*}
        p_1 &= \sum_{m\in \nat} \tau_S(\nu, n, a)\Big(\big\{ (w, m) \mid w\in \tau_R(\delta) \big\}\Big)\\
        &= \nu(\star)\cdot \dirac{\delta(a) = (\_, \top)} + \sum_{m\in \nat} \sum_{\nu'\in \supp(\nu)} \nu(\nu') \cdot \nu'\Big(\big\{ (w, m-n) \mid \delta(a) = (T, \_),\, w\in T\big\} \Big)\\
        &= \nu(\star)\cdot \dirac{\delta(a) = (\_, \top)} + \sum_{m\in \nat} \sum_{\nu'\in \supp(\nu)} \nu(\nu') \cdot \nu'\Big(\big\{ (w, m) \mid \delta(a) = (T, \_),\, w\in T\big\} \Big),\\
        r_1 &= \sum_{m\in \nat} m\cdot \tau_S(\nu, n, a)\Big(\big\{ (w, m) \mid w\in \tau_R(\delta) \big\}\Big)\\
        &= n\cdot \nu(\star)\cdot \dirac{\delta(a) = (\_, \top)} + \sum_{m\in \nat} \sum_{\nu'\in \supp(\nu)} m\cdot \nu(\nu') \cdot \nu'\Big(\big\{ (w, m-n) \mid \delta(a) = (T, \_),\, w\in T\big\} \Big)\\
        &= n\cdot \nu(\star)\cdot \dirac{\delta(a) = (\_, \top)} + \sum_{m\in \nat} \sum_{\nu'\in \supp(\nu)} (m+n)\cdot \nu(\nu') \cdot \nu'\Big(\big\{ (w, m) \mid \delta(a) = (T, \_),\, w\in T\big\} \Big),
    \end{align*}
    and $\tau_{S\otimes R} \circ \distributionfunc(q + \settorf)\times \nat\circ \lambda_{\subdistributionfunc(A^{+}\times \nat), \nempfaclang}(\nu, n, a, \delta) = (p_2, r_2)$ and $\lambda_{\subdistributionfunc(A^{+}\times \nat), \nempfaclang}(\nu, n, a, \delta) = (\mu, n)$, where 
    \begin{align*}
        p_2 &= \mu(\top) + \sum_{(\nu', L') \in \supp(\mu)} \mu(\nu', L')\cdot \pi_0\big(q(\nu', L')\big)\\
        &= \nu(\star)\cdot \dirac{\delta(a) = (\_, \top)} + \sum_{\nu'\in \supp(\nu)} \nu(\nu') \cdot \Big(\sum_{m\in \nat} \nu'\Big(\big\{ (w, m) \mid \delta(a) = (T, \_),\, w\in T\big\} \Big)\Big)\\
        &= \nu(\star)\cdot \dirac{\delta(a) = (\_, \top)} + \sum_{m\in \nat}\sum_{\nu'\in \supp(\nu)} \nu(\nu') \cdot  \nu'\Big(\big\{ (w, m) \mid \delta(a) = (T, \_),\, w\in T\big\} \Big),\\
        r_2 &= n\cdot \mu(\top) + \sum_{(\nu', L') \in \supp(\mu)} \mu(\nu', L')\cdot\Big( \pi_0\big(q(\nu', L')\big)\cdot n + \pi_1\big(q(\nu', L')\big)\Big)\\
        &= n\cdot \mu(\top) + \sum_{\nu' \in \supp(\nu)} \nu(\nu')\cdot\Big( \pi_0\big(q(\nu', L')\big)\cdot n + \pi_1\big(q(\nu', L')\big)\Big)\\
        &= n\cdot \mu(\top) + \sum_{\nu' \in \supp(\nu)} \nu(\nu')\cdot\Big( \sum_{m_1\in \nat}\nu'\Big(\big\{ (w, m_1) \mid \delta(a) = (T, \_),\, w\in T\big\}\Big)\cdot n\\
        &+ \Big( \sum_{m_2\in \nat}m_2\cdot \nu'\Big(\big\{ (w, m_2) \mid \delta(a) = (T, \_),\, w\in T\big\}\Big)\Big)\\
        &= n\cdot \mu(\top) + \sum_{\nu' \in \supp(\nu)} \nu(\nu')\cdot\Big( \sum_{m\in \nat} \big(m+n\big)\cdot \nu'\big(\big\{ (w, m) \mid \delta(a) = (T, \_),\, w\in T\big\}\big)\Big)\\
        &= n\cdot \nu(\star)\cdot \dirac{\delta(a) = (\_, \top)} + \sum_{m\in \nat} \sum_{\nu'\in \supp(\nu)} (m+n)\cdot \nu(\nu') \cdot \nu'\Big(\big\{ (w, m) \mid \delta(a) = (T, \_),\, w\in T\big\} \Big),
    \end{align*}

\end{proof}

\subsection{Costs Induced by Requirements}
\label{appendix:costReq}

\begin{definition}[modality for the reward machine]
    The modaility $\tau_S$ is given by $\tau_S\colon \big((B^{+}_{\bot})^{A^{+}}\times B\big)^A \rightarrow (B^{+}_{\bot})^{A^{+}}$, where 
    \begin{align*}
        \tau_S(\delta)(w) \defeq \begin{cases}
            \pi_2(\delta(a)) &\text{ if $w = a$,}\\
            \pi_2(\delta(a)) \cdot \pi_1(\delta(a))(w') &\text{ if $w = a\cdot w'$, and $\pi_1(\delta(a))(w') \in B^+$,}\\
            \bot &\text{ otherwise. }
        \end{cases}
    \end{align*}
\end{definition}

\begin{lemma}
    The predicate tranformer $\Phi\colon \big((B^{+}_{\bot})^{A^{+}}\big)^Y\rightarrow \big((B^{+}_{\bot})^{A^{+}}\big)^Y$ is $\omega$-continuous.  
\end{lemma}
\begin{proof}
    Let $(c_n)$ be an $\omega$-chain. 
    \begin{align*}
        &\Phi(\vee (c_n))(y)(w)\\
        &= \begin{cases}
            \pi_2(d(y)(a)) &\text{ if $w = a$,}\\
            \pi_2(d(y)(a)) \cdot \Big(\vee (c_n)\big(\pi_1(d(y)(a))\big)(w') \Big) &\text{ if $w = a\cdot w'$, and $\vee (c_n)\big(\pi_1(d(y)(a))\big)(w') \in B^+$,}\\
            \bot &\text{ otherwise. }
        \end{cases}\\
        &= \begin{cases}
            \pi_2(d(y)(a)) &\text{ if $w = a$,}\\
            \pi_2(d(y)(a)) \cdot  c_m\big(\pi_1(d(y)(a))\big)(w') &\text{ if $w = a\cdot w'$, and $\exists m. c_m\big(\pi_1(d(y)(a))\big)(w') \in B^+$,}\\
            \bot &\text{ otherwise. }
        \end{cases}\\
        &= \vee\big(\Phi(c_n)\big)(y)(w).\\
    \end{align*} 
\end{proof}

\begin{definition}[inference for reward machines and bounds]
    The inference $q$ is given by $q\colon (B^{+}_{\bot})^{A^{+}}\times\nempfaclang[B]  \rightarrow \nempfaclang$ by 
    \begin{align*}
        q(f, L) \defeq \{ w\in A^{+}\mid f(w)\in L\}. 
    \end{align*}
\end{definition}

\begin{lemma}
    The coalgebraic product $d_1\otimes_{\alpha} d_2$ meets the correctness criterion. 
\end{lemma}
\begin{proof}
    The query $q$ preserves the least elements and $\omega$-chains, and the following diagram commutes: 
    \begin{equation*}
        \adjustbox{scale=0.9,center}{%
        \begin{tikzcd}
            \big((B^{+}_{\bot})^{A^+}\times B\big)^A\times (\nempfaclang[B] \times \settorf)^B \arrow[rrr,"\modsys{\tau}\times \modspec{\tau}"] \arrow[d,"\alpha_{(B^{+}_{\bot})^{A^+},\, \nempfaclang[B]}"]
                    && 
                    &(B^{+}_{\bot})^{A^+}\times\nempfaclang[B]\arrow[d,"q"]\\
                    \dfautomata{(B^{+}_{\bot})^{A^+}\times \nempfaclang[B]}\arrow[rr,"(q\times \settorf)^A"]
                    && \dfautomata{\nempfaclang[A]} \arrow[r, "\modprod{\tau}"]
                    & \nempfaclang
        \end{tikzcd}
        }
        \end{equation*}

        \begin{align*}
            &\big(q\circ \modsys{\tau}\times \modspec{\tau}\big)(\delta_1, \delta_2)\\
            &= \{w\mid w\in A^{+}, \modsys{\tau}(\delta_1)(w)\in B^{+},\text{ and }\modsys{\tau}(\delta_1)\in \modspec{\tau}(\delta_2)\}\\
            &= \{a\mid (\_, b) = \delta_1(a),\, (\_, \top) = \delta_2(b)\}\bigcup \{a\cdot w_1\mid (w_2, b) = \delta_1(a),\, (L, \_) = \delta_2(b),\, w_2\in L\}\\
            &\big(\modprod{\tau}\circ (q\times \settorf)^A\circ \alpha_{(B^{+}_{\bot})^{A^+},\, \nempfaclang[B]}\big)(\delta_1, \delta_2) \\
            &= \{ a\mid (\_, b) = \delta_1(a),\, (\_, \top) = \delta_2(b)\}\bigcup \{a\cdot w_1\mid (w_2, b) = \delta_1(a),\, (L, \_) = \delta_2(b),\, w_2\in L\}
        \end{align*}
\end{proof}

\subsection{MCs without Target States and Safety Property}
\label{sec:proof_inf_trace}

\begin{definition}
    The modality $\tau_S$ defined in~\cref{def:modalityMCs} is well-defined.
\end{definition}
\begin{proof}
    Let $\nu\in \distributionfunc(\Omega)$, $a\in A$, and $m\defeq \min\big(\{n\mid \sigma\in \supp(\nu)\text{ and }\sigma\in \distributionfunc(A^n)\}, \omega\big)$.
    It suffices to prove that the modality $\tau_S$ is well-defined when $m = \omega$. 
    For each $m, n\in \nat$, we prove that $m <  n$ implies $\sigma_m \prec \sigma_n$. 
    For any $w \in A^m$, 
    \begin{align*}
        \sigma_m(w) &\defeq \begin{cases}
            \sum_{\sigma\in \supp(\nu)} \nu(\sigma)\cdot \mrg{\sigma}{m-1}(w') &\text{ if $w = a\cdot w'$,}\\
            0 &\text{ if $a\not \in\pref{w}$.}
        \end{cases}\\
        &=\begin{cases}
            \sum_{\sigma\in \supp(\nu)} \nu(\sigma)\cdot \sigma(\cyset{w'}) &\text{ if $w = a\cdot w'$,}\\
            0 &\text{ if $a\not \in\pref{w}$.}
        \end{cases}\\
        \mrg{\sigma_n}{m}(w) &\defeq \sum_{w_2\in A^n\land w\in \pref{w_2}} \sigma_n(w_2)\\
        &=\sum_{a\cdot w_3\in A^n\land w\in \pref{a\cdot w_3}}\sum_{\sigma\in \supp(\nu)} \nu(\sigma)\cdot \sigma(\cyset{w_3})\\
        &=\begin{cases}
            \sum_{\sigma\in \supp(\nu)} \nu(\sigma)\cdot \sigma(\cyset{w'}) &\text{ if $w = a\cdot w'$,}\\
            0 &\text{ if $a\not \in\pref{w}$.}
        \end{cases}\\
    \end{align*}
\end{proof}

\begin{lemma}
    The predicate transformer $\predtran{\tau_S, c}$ on $\Omega^X$ is $\omega$-continuous, where the semantic domain $\Omega$ is defined in~\cref{def:truth-valueMC}. 
\end{lemma}
\begin{proof}
Let $u, v\in \Omega^X$ and $u\preceq v$. 
For each $x\in X$ and $c(x) = \big(P(x, \_), l(x)\big)$, let $m = \min \{u(x') \mid x'\in \sup(P(x, \_))\}$ and $n = \min \{u(x') \mid x'\in \sup(P(x, \_))\}$.
Since $u\preceq v$, $m\leq n$, and $\predtran{\tau_S, c}(u)(x)\preceq \predtran{\tau_S, c}(v)(x)$. 

Let $(u_n)_{n\in \nat}$ be an $\omega$-chain.
For each $x\in X$, we assume that there is $x_1\in \supp\big(P(x, \_)\big)$ such that $u_n(x_1)\in \bigcup_{m\in \nat} \distributionfunc(A^m)$ for any $n\in \nat$, and $\big(\bigvee (u_n)_{n\in \nat}\big)(x_2)\in \probmeasure(\salg{A})$ for any $x_2\in \supp(P(x, \_))$; 
otherwise, it is trivial that the predicate transformer $\predtran{\tau_S, c}$ preserves the join. 
Under the assupmtion, there is an $\omega$-chain $(v_n)_{n\in \nat}$ such that $v_n(x_1)\in  \distributionfunc(A^n)$ for any $x_1\in \supp\big(P(x, \_)\big)$, and $\predtran{\tau_S, c}(\bigvee_{n\in \nat} v_n)(x) = \predtran{\tau_S, c}(\bigvee_{n\in \nat}u_n)(x)$ and $\bigvee_{n\in \nat}\predtran{\tau_S, c}(u_n)(x) = \bigvee_{n\in \nat}\predtran{\tau_S, c}(v_n)(x)$. 
Then, it suffices to prove that $\predtran{\tau_S, c}(\bigvee_{n\in \nat}v_n)(x) = \bigvee_{n\in \nat}\predtran{\tau_S, c}(v_n)(x)$. This equality holds by the definition of the modality $\tau_S$. 
\end{proof}

\begin{lemma}
    The query $q$ is $\omega$-continuous. 
\end{lemma}
\begin{proof}
    Let $(\sigma_n)_{n\in \nat}$ and $T_n$ be $\omega$-chains. 
    We assume that $\bigvee_{n\in \nat} \sigma_n\in \probmeasure(\salg{A})$; otherwise it is easy to prove that the query is $\omega$-continuous. 
    We prove the statement by the following calculation: 
    \begin{align*}
        \bigvee_{n\in \nat} \sigma_n\big((\bigvee_{m\in \nat}T_m)^{\omega}\big) &= \lim_{k\rightarrow \infty}\sum_{i\leq k}\sigma_{n_i} (\bigvee_{m\in \nat}T_m)^i\\
        &= \lim_{k\rightarrow \infty}\lim_{m\rightarrow \infty }\sum_{i\leq k}\sigma_{n_i} (T_m)^i\\
        &= \lim_{k\rightarrow \infty}\sum_{i\leq k}\sigma_{n_i} (T_k)^i\\
        &= \bigvee_{n\in \nat} q(\sigma_n, T_n).
    \end{align*}
\end{proof}
    
\begin{lemma}
    The data commutes the diagram shown in~\cref{prop:weaker_correctness_criterion}.  
\end{lemma}
\begin{proof}
    Let $(\nu, a, \delta)\in\mcinffunc{X}\times\dfautomata{Y}$. 
    We write $(\nu', a, \delta')$ for $(F_S(u_{S, k})(\nu, a), F_R(u_{R, k})(\delta))$. 
    By definition of $u_{S, k}$, we know that for any $\sigma\in \supp(\nu')$, $\sigma\in \distributionfunc(A^k)$. 
    This is crucial for the proof.

    If $\delta'(a) = (\_, \top)$, then $\big(q \circ \modsys{\tau}\times \modspec{\tau}\big)(\nu', a, \delta') = 1$ and $\big(\modprod{\tau}\circ \dmcflagfunc{q}\circ \lambda_{\Omega,\, \nempfaclang}\big)(\nu', a, \delta') = 1$. 
    Suppose that $\delta'(a) = (\_, \bot)$. The following equations hold. 
    \begin{align*}
        & \big(q \circ \modsys{\tau}\times \modspec{\tau}\big)(\nu', a, \delta') \\
        &= \sum_{i\in [1, k+1]}\mrg{\tau_S(\nu', a)}{i}(\tau_R(\delta')^i)\\
        &= \sum_{i\in [1, k+1]}\sum_{\sigma\in \supp(\nu')}\nu'(\sigma)\cdot \mrg{\sigma}{i-1}(\{w \mid a\cdot w\in\modspec{\tau}(\delta')\}^{i-1})\\
        &= \sum_{\sigma\in \supp(\nu')}\nu'(\sigma)\cdot \Big(\sum_{i\in [1, k]} \mrg{\sigma}{i}(\{w \mid a\cdot w\in\modspec{\tau}(\delta')\}^{i})\Big)\\
        &= \sum_{\sigma\in \supp(\nu')}\nu'(\sigma) \cdot q(\sigma, \{w \mid a\cdot w\in\modspec{\tau}(\delta')\})\\
        &= \sum_{\sigma\in \supp(\nu')}\nu'(\sigma) \cdot q\Big(\sigma, \pi_1\big(\delta'(a)\big)\Big) \\
        &\big(\modprod{\tau}\circ \dmcflagfunc{q}\circ \lambda_{\Omega,\, \nempfaclang}\big)(\nu', a, \delta')\\
        &= \sum_{r\in \probinterval} r\cdot\dmcflagfunc{q}\circ \lambda_{\Omega,\, \nempfaclang}(\nu', a, \delta')(r) + \dmcflagfunc{q}\circ \lambda_{\Omega,\, \nempfaclang}(\nu', a, \delta')(\star) \\
        &= \sum_{\sigma\in \supp(\nu')} \nu'(\sigma)\cdot  q\Big(\sigma, \pi_1\big(\delta'(a)\big)\Big)\\
    \end{align*}

\end{proof}
\subsection{A Translation of Inferences on MCs}
\label{sec:translation}

\begin{proposition}
    \label{prop:translationMC}
    Let $c\defeq (P, l)\colon X\rightarrow \mcfunc{X}$ be a finite branching MC with the target state and 
    $d\colon Y\rightarrow \dfautomata{Y}$ be a DFA. Consider the MC $c'\defeq (P', l')\colon X+\{\star\}\rightarrow \mcinffunc[(A+\{\star_A\})]{X+\{\star\}}$ and the DFA $d' \colon (Y\times \settorf)\rightarrow \dfautomata[(A+\{\star_{A}\})]{(Y\times \settorf)}$ such that 
    \begin{align*}
        P'(z, z') &\defeq \begin{cases}
            P(z, z') &\text{ if $z, z'\in X$,}\\
            P(z, \star) &\text{ if $z\in X,\, z' = \star,$}\\
            1 &\text{ if $z = z' = \star$,}\\
            0 &\text{ otherwise. }    
        \end{cases}
        &l'(z) \defeq \begin{cases}
            l(z) &\text{ if $z \in X$,}\\
            \star_A &\text{ if $z = \star$,} 
        \end{cases}\\
        d'(z)(a) &\defeq \begin{cases}
            \big((y', b),\bot\big) &\text{ if $a\in A,\,z = (y, \_),$ and $ d(y)(a) = (y', b)$, }\\
            \big((y, b),b\big) &\text{ if $a = \star_A$, and $z = (y, b)$, }
        \end{cases}
            &
    \end{align*}
    where the new alphabet $\star_A$ is for the sink state $\star$. 
    Then, the quantitative inference on $(c, d)$ can be reduced to the one on $(c', d')$, that is, $q\big( \lfpoint[x]{c},\,\lfpoint[y]{d}\big) = q\big( \lfpoint[x]{c'},\,\lfpoint[(y, b)]{d'}\big)$, where $b$ can be $\top$ or $\bot$. 
\end{proposition}
\begin{proof}[Proof Sketch]
    Since the product constructions for $(c, d)$ and $(c', d')$ are correct, it suffices to show that
    there is a bijective correspondence between the paths that reach the target state on products while preserving the path probability. 
    This is easy to prove by the construction. 
\end{proof}

\subsection{Weighted transition systems and NFAs}
\label{sec:proof_weight_nfas}
\begin{proposition}
The query $q$ given in~\cref{def:query_wts_nfas} is $\omega$-continuous and preserves the least elements, and makes the following diagram commute: 
 \begin{equation*}
 \small
\begin{tikzcd}
            \ltsfuncsemiring{\powersetfunc(A^{+}\times \nat)}{\nat}\times \exepnfautomata{\nempfaclang}\arrow[rr,"\modsys{\tau}\times \modspec{\tau}"] \arrow[d,"\lambda_{\powersetfunc(A^{+}\times \nat), \nempfaclang}"]
            & 
            &\powersetfunc(A^{+}\times \nat)\times \nempfaclang \arrow[d,"q"]\\
 \suptlflagfunc{\powersetfunc(A^{+}\times \nat)\times \nempfaclang}{\nat} \arrow[r,"\suptlflagfunc{q}{\nat}"]
            & \suptlflagfunc{\natinf}{\nat} \arrow[r, "\modprod{\tau}"]
            & \natinf
\end{tikzcd}
\end{equation*}
\end{proposition}
\begin{proof}
Firstly, we prove the query $q$ is $\omega$-continuous. 
Let  $(S_i)_{i\in \nat}$ and $(T_j)_{j\in \nat}$ be $\omega$-chains on $\powersetfunc(A^{+}\times \nat)$ and $\nempfaclang$. We prove the following equation:
\begin{align*}
    \min \Big\{m \ \Big| \ w\in \bigvee_{j\in \nat} T_j,\, (w, m)\in \bigvee_{i\in \nat} S_i \Big\} = \bigvee_{k\in \nat} \min \Big\{m \ \Big| \ w\in T_k,\, (w, m)\in  S_k \Big\}.
\end{align*}

The following inequality is trivial, thus we prove the opposite inequality: 
\begin{align*}
    \min \Big\{m \ \Big| \ w\in \bigvee_{j\in \nat} T_j,\, (w, m)\in \bigvee_{i\in \nat} S_i \Big\} \geq_o \bigvee_{k\in \nat} \min \Big\{m \ \Big| \ w\in T_k,\, (w, m)\in  S_k \Big\}.
\end{align*}

Let $w\in \bigvee_{j\in \nat} T_j$ and $(w, m)\in \bigvee_{i\in \nat} S_i$ such that $m = \min \Big\{m \ \Big| \ w\in \bigvee_{j\in \nat} T_j,\, (w, m)\in \bigvee_{i\in \nat} S_i \Big\}$. Then, there is $j\in \nat$ and $i\in \nat$ such that $w\in T_j$ and $(w, m)\in S_i$. If we take $k \defeq \max(i, j)$, then $w\in T_k$ and $(w, m)\in S_k$. Thus, we obtain the following inequality: 
\begin{align*}
    \min \Big\{m \ \Big| \ w\in \bigvee_{j\in \nat} T_j,\, (w, m)\in \bigvee_{i\in \nat} S_i \Big\} =  \min \Big\{m \ \Big| \ w\in  T_k,\, (w, m)\in S_k \Big\} \leq_o \bigvee_{k\in \nat} \min \Big\{m \ \Big| \ w\in T_k,\, (w, m)\in  S_k \Big\}.
\end{align*}

Preserving the least elements is trivial, thus we finish the proof by showing the diagram commutes. 
{\small
\begin{align*}
    &\big(q \circ \modsys{\tau}\times \modspec{\tau}\big)(S, \delta) \defeq \min\big\{ m \ \big| \ w\in \modspec{\tau}(\delta),\, (w, m)\in \modsys{\tau}(S) \big\}\\
    &= \min \Big(\big\{\ m \ \big| \ (\_, \top) \in \delta(a),\, (\star,a, m)\in S  \big\}\cup\big\{\ m+n \ \big| \ w\in A^{+},\,(T_2, \_) \in \delta(a),\, w\in T_2 ,\, (T_1,a, m)\in S,\, (w, n)\in T_1  \big\}   \Big)\\
    &\Big(\tau \circ \powersetfunc\big((q+\settorf)\times \nat\big)\circ \lambda_{\powersetfunc(A^{+}\times \nat)\times \nempfaclang}\Big)(S, \delta)\\
    &= \Big(\tau \circ \powersetfunc\big((q+\settorf)\times \nat\big)\Big)\big(\big\{(b,m) \ \big| \ (\star, a, m)\in S \text{ and }(\_, b)\in\delta(a)\big\}\cup \big\{ (T_1, T_2,  m) \ \big| \ (T_1, a, m)\in S,\, (T_2,\_)\in \delta(a)\big\}\big)\\
    &= \tau\big(\big\{(b,m) \ \big| \ (\star, a, m)\in S \text{ and }(\_, b)\in\delta(a)\big\}\cup \big\{ (q(T_1, T_2), m) \mid (T_1, a, m)\in S,\, (T_2,\_)\in \delta(a) \big\}\big)\\
    &= \min \Big(\big\{\ m \ \big| \ (\_, \top) \in \delta(a),\, (\star,a, m)\in S  \big\}\cup\big\{\ m+n \ \big| \ w\in A^{+},\,(T_2, \_) \in \delta(a),\, w\in T_2 ,\, (T_1,a, m)\in S,\, (w, n)\in T_1  \big\}   \Big).
\end{align*}
}
\end{proof}

\subsection{Weighted transition systems and WMAs}
\label{sec:proof_weight_WMA}

\begin{proposition}
    The query $q$ given in~\cref{def:query_wts_WMAs} is $\omega$-continuous and preserves the least elements, and makes the following diagram commute: 
     \begin{equation*}
     \tiny
    \begin{tikzcd}
                \ltsfuncsemiring{\powersetfunc(A^{+}\times \nat)}{\nat}\times \wmealyfunc{\powersetfunc(A^{+}\times \nat)}\arrow[rr,"\modsys{\tau}\times \modspec{\tau}"] \arrow[d,"\lambda_{\powersetfunc(A^{+}\times \nat), \powersetfunc(A^{+}\times \nat)}"]
                & 
                &\powersetfunc(A^{+}\times \nat)\times\powersetfunc(A^{+}\times \nat) \arrow[d,"q"]\\
     \suptlflagfunc{\powersetfunc(A^{+}\times \nat)\times \powersetfunc(A^{+}\times \nat)}{\nat} \arrow[r,"\suptlflagfunc{q}{\nat}"]
                & \suptlflagfunc{\natinf}{\nat} \arrow[r, "\modprod{\tau}"]
                & \natinf
    \end{tikzcd}
    \end{equation*}
\end{proposition}
\begin{proof}
    Firstly, we prove the query $q$ is $\omega$-continuous. 
    Let  $(S_i)_{i\in \nat}$ and $(T_j)_{j\in \nat}$ be $\omega$-chains on $\powersetfunc(A^{+}\times \nat)$. We prove the following equation:
    \begin{align*}
        \min \Big\{m+n \ \Big| \ (w, m)\in \bigvee_{j\in \nat} S_j,\, (w, n)\in \bigvee_{i\in \nat} T_i \Big\} = \bigvee_{k\in \nat} \min \Big\{ m+n \ \Big| \ (w, m)\in S_k,\, (w, n)\in  T_k \Big\}.
    \end{align*}
    
    The following inequality is trivial, thus we prove the opposite inequality: 
    \begin{align*}
        \min \Big\{m+n \ \Big| \ (w, m)\in \bigvee_{j\in \nat} S_j,\, (w, n)\in \bigvee_{i\in \nat} T_i \Big\} \geq_o \bigvee_{k\in \nat} \min \Big\{m+n \ \Big| \ (w, m)\in S_k,\, (w, n)\in  T_k \Big\}.
    \end{align*}
    
    Let $(w, m)\in \bigvee_{j\in \nat} S_j$ and $(w, n)\in \bigvee_{i\in \nat} T_i$ such that $m + n = \min \Big\{m + n\ \Big| \ (w, m)\in \bigvee_{j\in \nat} S_j,\, (w, n)\in \bigvee_{i\in \nat} T_i \Big\}$. Then, there is $j\in \nat$ and $i\in \nat$ such that $(w, m)\in S_j$ and $(w, n)\in T_i$. If we take $k \defeq \max(i, j)$, then $(w, m)\in S_k$ and $(w, n)\in T_k$. Thus, we obtain the following inequality: 
    \begin{align*}
        \min \Big\{m + n \ \Big| \ (w, m)\in \bigvee_{j\in \nat} S_j,\, (w, n)\in \bigvee_{i\in \nat} T_i \Big\} &=  \min \Big\{m + n \ \Big| \ (w, m)\in  S_k,\, (w, n)\in T_k \Big\}\\
        & \leq_o \bigvee_{k\in \nat} \min \Big\{m + n \ \Big| \ (w, m)\in S_k,\, (w, n)\in  T_k \Big\}.
    \end{align*}
    
    Preserving the least elements is trivial, thus we finish the proof by showing the diagram commutes. 
    {\small
    \begin{align*}
        &\big(q \circ \modsys{\tau}\times \modspec{\tau}\big)(S, \delta) \defeq \min\big\{ m + n \ \big| \ (w, m)\in \modsys{\tau}(S),\, (w, n)\in \modspec{\tau}(\delta) \big\}\\
        &= \min \Big(\big\{\ m + n \ \big| \ (\star,a, m)\in S, (\_, \top, n) \in \delta(a)  \big\}\\
        &\bigcup\big\{\ m_1 + m_2 + n_1 + n_2 \ \big| \ w\in A^{+},\,(T_2, \_, n_1) \in \delta(a),\, (w, n_2)\in T_2 ,\, (T_1,a, m_1)\in S,\, (w, m_2)\in T_1  \big\}   \Big)\\
        &\Big(\tau \circ \powersetfunc\big((q+\settorf)\times \nat\big)\circ \lambda_{\powersetfunc(A^{+}\times \nat)\times\powersetfunc(A^{+}\times \nat)}\Big)(S, \delta)\\
        &= \Big(\tau \circ \powersetfunc\big((q+\settorf)\times \nat\big)\Big)\Big(\big\{(b,m+n) \ \big| \ (\star, a, m)\in S \text{ and }(\_, b, m)\in\delta(a)\big\}\\
        &\bigcup \big\{ (T_1, T_2,  m+n) \ \big| \ (T_1, a, m)\in S,\, (T_2,\_, n)\in \delta(a)\big\}\Big)\\
        &= \tau\big(\big\{(b,m + n) \ \big| \ (\star, a, m)\in S \text{ and }(\_, b, n)\in\delta(a)\big\}\cup \big\{ (q(T_1, T_2), m+n) \mid (T_1, a, m)\in S,\, (T_2,\_, n)\in \delta(a) \big\}\big)\\
        &= \min \Big(\big\{\ m + n \ \big| \ (\_, \top, n) \in \delta(a),\, (\star,a, m)\in S  \big\}\\
        &\bigcup\big\{\ m_1+n_1+m_2+n_2 \ \big| \ w\in A^{+},\,(T_2, \_, n_1) \in \delta(a),\, (w, n_2)\in T_2 ,\, (T_1,a, m_1)\in S,\, (w, m_2)\in T_1  \big\}   \Big).
    \end{align*}
    }
\end{proof}

\section{Fibrational Perspectives on Coalgebraic Products}
\label{sec:lifting} 

\begin{proposition}[lifting functor $\otimes$ via a distributive law $\lambda$ ]
\label{prop:lift_product}
Let $\lambda$ be a distributive law from functors $F_S$ and $F_R$ to a functor $F_{S\otimes R}$ on $\sets$. The distributive law $\lambda$ gives the lifting $\otimes: \coalg{F_S}\times\coalg{F_R}\rightarrow \coalg{F_{S\otimes R}} $ of the cartesian product $\times$ in the category $\sets$ of sets. To put it precisely, the coalgebraic product construction $\otimes$ with $\lambda$ forms the functor $\otimes\colon \coalg{F_S}\times\coalg{F_R}\rightarrow \coalg{F_{S\otimes R}}$ with commuting the following diagram:
\begin{equation*}
\begin{tikzcd}
    \coalg{F_S}\times\coalg{F_R} \arrow[r, "\otimes"] \arrow[d, "U_S\times U_R"]& \coalg{F} \arrow[d, "U_{S\otimes R}"]\\
    \sets \times \sets  \arrow[r, "\times"] & \sets
\end{tikzcd}
\end{equation*}
where $U_S, U_R$, and $U_{S\otimes R}$ are forgetful functors.
\end{proposition}
\begin{proof}
We prove the coalgebraic product construction forms the functor; proving the diagram commutes is trivial. 
For arrows $f\in \coalg{F_S}(c_1, c_2)$ and $g\in \coalg{F_R}(d_1, d_2)$, the arrow $f\otimes g \in\coalg{F_{S\otimes R}}(c_1\otimes_{\lambda} d_1, c_2\otimes_{\lambda} d_2) $ is defined by $f\otimes g \defeq f\times g$. This is well-defined, i.e., the following diagram commutes by the naturality of $\lambda$: 
\begin{equation*}
\begin{tikzcd}
    F_{S\otimes R}(X_1\times Y_1) \arrow[rr, "F_{S\otimes R}(f\times g)"] && F(X_2\times Y_2)\\
    F_S(X_1) \times F_R(Y_1) \arrow[u, "\lambda_{X_1, Y_1}"] \arrow[rr, "F_S(f)\times F_R(g)"] && F_S(X_2)\times F_R(Y_2) 
    \arrow[u, "\lambda_{X_2, Y_2}"]\\
    X_1\times Y_1  \arrow[u, "c_1\times d_1"] \arrow[rr, "f\times g"]&& X_2\times Y_2\arrow[u, "c_2\times d_2"]
\end{tikzcd}
\end{equation*}
Preserving identities and sequential compositions are easy to prove. 
\end{proof}

Fibrations have been studied for characterising  coinductive predicates~\cite{HermidaJ98,HasuoKC18} in a coalgebra. In this paper, we only define the predicate fibration without introducing fibration theory~\cite{Jacobs2001}.

\begin{definition}[slice category, and predicate fibration]
\label{def:slicecat_predfib}
Let $\Omega$ be a set. The \emph{slice category} $\slice{\sets}{\Omega}$ is the category whose object is a function $u\colon X\rightarrow \Omega$ and arrow from $u_1\colon X\rightarrow \Omega$ to $u_2\colon Y\rightarrow \Omega$ is a function $f\colon X\rightarrow Y$ such that $u_2\circ f = u_1$. The \emph{predicate fibration} $p\colon\slice{\sets}{\Omega}\rightarrow \sets$ is the forgetful functor, in particular, for a function $u\colon X\rightarrow \Omega$, $p(u)$ is given by  $p(u)\defeq X$. 
\end{definition}

There is a well-establish construction of the lifting $\liftfunc{F}{\tau}$ of the functor $F$ on $\sets$ along the predicate fibration $p\colon\slice{X}{\Omega}\rightarrow \Omega$ with the \emph{modality} $\tau\colon F(\Omega)\rightarrow \Omega$. We construct the predicate transformer $\predtran{\tau,c}$ from the lifting $\liftfunc{F}{\tau}$ and a coalgebra $c\colon X\rightarrow F(X)$.

\begin{definition}[lifting, and predicate transformer]
\label{def:lifting_predtran}
Assume the following: (i) a functor $F$ on $\sets$,
(ii) a set $\Omega$, (iii) an arrow $\tau\colon F(\Omega)\rightarrow \Omega$.

The \emph{lifting} $\liftfunc{F}{\tau}$ is the functor on $\slice{\sets}{\Omega}$ given by $\liftfunc{F}{\tau}(u)\defeq \tau\circ F(u)$ and $\liftfunc{F}{\tau}(g)\defeq F(g)$ for an object $u\colon X\rightarrow \Omega$ and an arrow $g\colon u_1\rightarrow u_2$.

In addition, let $c\colon X\rightarrow F(X)$ be a coalgebra. The predicate transformer $\predtran{\tau,c}$ on $\sets(X, \Omega)$ is given by $\predtran{\tau,c}(u) \defeq \liftfunc{F}{\tau}(u)\circ c$. 
\end{definition}

In fibrational point of view, the commutative diagram in~\cref{thm:correctness_product} is a sufficient condition for the lifting $\liftdist{\lambda}$ of a distributive law $\lambda$ along the predicate fibration as illustrated in~\cref{fig:liftingNaturalTransformation}. 
\begin{proposition}[lifting   $\liftdist{\lambda}$ of a distributive law $\lambda$]
\label{prop:lifting_of_distributive_law}
Assume the diagram in~\cref{thm:correctness_product} commutes. 
Then, there is the lifting $\liftdist{\lambda}\colon \liftdist{q}\circ \liftdist{\times}\circ \big(\liftfunc{F_S}{\tau_S} \times \liftfunc{F_R}{\tau_R}\big) \Rightarrow  \liftfunc{F_{S\otimes R}}{\tau_{S\otimes R}}\circ\liftdist{q}\circ \liftdist{\times}$ of the distributive law $\lambda\colon \times \circ F_S\times F_R\Rightarrow F_{S\otimes R}\circ \times $ shown in~\cref{fig:liftingNaturalTransformation}, where the functor $\liftdist{q}\colon \Frac{\sets}{\Omega_S\times \Omega_R} \rightarrow \Frac{\sets}{\Omega_{S\otimes R}}$ is given by $\liftdist{q}(u) \defeq q \circ u $, and the functor $\liftdist{\times}\colon \Frac{\sets}{\Omega_S}\times  \Frac{\sets}{\Omega_R} \rightarrow  \Frac{\sets}{\Omega_S\times \Omega_R} $ is given by $\liftdist{\times}(u_S, u_R) \defeq  u_S\times u_R$. 
\qed
\end{proposition}


\begin{figure}
\centering
\begin{tikzcd}
     (\Frac{\sets}{\Omega_S})\times   (\Frac{\sets}{\Omega_R})
      \arrow[rrrr, "\widetilde{q}\circ \widetilde{\times}\circ \big(\liftfunc{F_S}{\tau_S} \times \liftfunc{F_R}{\tau_R}\big)"{name=U, above}, bend left=20]
      \arrow[rrrr, bend right=20, "\liftfunc{F_{S\otimes R}}{\tau_{S\otimes R}}\circ\widetilde{q}\circ \widetilde{\times}"{name=D, below}] \arrow[ddd, "p_S\times p_R"] 
     &
     &
     &
     &\Frac{\sets}{\Omega_{S\otimes R}} 
     \arrow[Rightarrow, "\liftdist{\lambda}", from=U, to=D] \arrow[ddd, "p_{S\otimes R}"] \\
     &
     &
     &
     &
     \\
     &
     &
     &
     &
     \\
     \sets\times \sets \arrow[rrrr, "\times \circ F_S \times F_R"{name=A, above}, bend left=20]
      \arrow[rrrr, bend right=20, "F_{S\otimes R}\circ \times "{name=B, below}] 
     &
     &
     & 
     & \sets \arrow[Rightarrow, "\lambda", from=A, to=B]
\end{tikzcd}
\caption{The lifting $\liftdist{\lambda}$ of the distributive law $\lambda$}
\label{fig:liftingNaturalTransformation}
\end{figure}

\else 
\fi

\end{document}